\documentclass[paper=letter, fontsize=10pt,leqno]{scrartcl}	% Article class of KOMA-script 
\usepackage[T1]{fontenc}
\usepackage{fourier}

\usepackage[english]{babel}															% English language/hyphenation
\usepackage{graphicx,amsmath,amsfonts,amsthm}										% Math packages
\usepackage{url}
\usepackage{amssymb,bbm, MnSymbol}
\usepackage[dvips]{color}
\usepackage{eucal}

\usepackage{sectsty}												% Custom sectioning (see below)
\allsectionsfont{\centering \normalfont\scshape}	% Change font of al section commands

\usepackage{fancyhdr}
\newtheorem{theorem}{Theorem}
\newtheorem{assumption}{Assumption}
\newtheorem{lemma}{Lemma}
\newtheorem{proposition}{Proposition}
\newtheorem{corollary}{Corollary}
\newtheorem{definition}{Definition}

\numberwithin{equation}{section}		% Equationnumbering: section.eq#
\numberwithin{figure}{section}			% Figurenumbering: section.fig#
\numberwithin{table}{section}				% Tablenumbering: section.tab#

\title{
		\vspace{-1in} 	
		\usefont{OT1}{bch}{b}{n}
		\normalfont \normalsize \textsc{} \\ [25pt]
		\huge Multiple scattering in random mechanical systems and  diffusion approximation \\
}

\author{\normalfont \large 
Renato Feres\footnote{Washington University, Department of Mathematics, Campus Box 1146, St. Louis, MO 63130},\  \ Jasmine Ng\footnotemark[1],\ \  Hong-Kun Zhang\footnote{Department of Mathematics and Statistics, University of Massachusetts, Amherst, MA 01003, USA}
}

\date{\normalfont  \large \today}

\begin{document}

\maketitle
\vspace{-0.2in}
\begin{abstract}
\begin{center}Abstract\end{center}
This paper is concerned with   stochastic   processes that 
model multiple (or iterated) scattering in    classical mechanical systems  
 of {\em billiard type},  defined below. 
 From  a given (deterministic) system of billiard type,
a  random process with transition probabilities operator $P$  is
introduced  by assuming that some of the
dynamical variables are random with  prescribed probability distributions. 
Of particular interest are systems with weak scattering, which are associated to parametric families of operators $P_h$, depending on 
a geometric or mechanical   parameter $h$, that approaches the identity 
as $h$ goes to $0$.  It is shown that 
$(P_h-I)/h$ converges for small $h$ to 
  a second order elliptic differential operator $\mathcal{L}$ on compactly supported functions 
  and  that the Markov chain  process associated to $P_h$ converges to a diffusion  with infinitesimal generator 
$\mathcal{L}$.
Both $P_h$ and $\mathcal{L}$ are
self-adjoint (densely) defined on the space $L^2(\mathbb{H},\eta)$ of square-integrable functions over 
the (lower) half-space $\mathbb{H}$   in $\mathbb{R}^m$,  where
$\eta$ is  a stationary measure. This  measure's density  is  either
(post-collision) Maxwell-Boltzmann distribution  or Knudsen cosine law, and the random processes with infinitesimal generator $\mathcal{L}$
respectively correspond to what we call {\em MB diffusion} and (generalized) {\em Legendre diffusion}. 
Concrete examples of simple mechanical systems  are given and illustrated by numerically simulating  the random  processes. 
 \end{abstract}

\section{Introduction}
The purpose of this section is to explain informally the nature of the results
that will be stated in detail and   greater generality  in the course of  the paper.

A type of idealized multi-scattering experiment is depicted in Figure \ref{multiscattering}.
The figure represents the flight of a molecule between two parallel solid plates.  At each collision,
the molecule impinges on the surface of a plate  with a velocity $v$ and, after interacting with the surface in
some way (which will be explicitly described by a mechanical   model),
it scatters away with a post-collision velocity $V$. The single scattering event $v\mapsto V$, for 
 some specified  molecule-surface  interaction model, is  given  by  a random map in the following
sense. Let $\mathbb{H}$ denote  the half-space of vectors $v=(v_1, v_2, v_3)$ with negative third component. 
It is convenient to also regard the scattered velocity $V$   as a vector in $\mathbb{H}$  by
identifying  vectors  that  differ only by the sign  of their  third component.
A scattering event is then represented by a map from  $\mathbb{H}$ into the space of probability
measures on $\mathbb{H}$, which we call  for now  the {\em scattering map}\,\!\! ;  the probability
measure associated to $v$ is the law of the random variable $V$.  Thus
the scattering map encodes the ``microscopic'' mechanism of molecule-surface interaction
in the form of a random map, whose iteration provides the   information about velocities needed to
determine  the  sample trajectories of the molecule.

The mechanical-geometric interaction models specifying the scattering map will be limited in this paper to
what we call a mechanical system of {\em billiard type}. Essentially, it is a conservative classical mechanical
system without ``soft'' potentials. Interactions between moving masses (comprising
the ``wall sub-system'' and the ``molecule sub-system,'' using the language of \cite{scott}) are  billiard-like elastic collisions.

 \vspace{0.1in}
\begin{figure}[htbp]
\begin{center}
%\epsfile{file=bundle.eps,scale=0.8}
\includegraphics[width=3in]{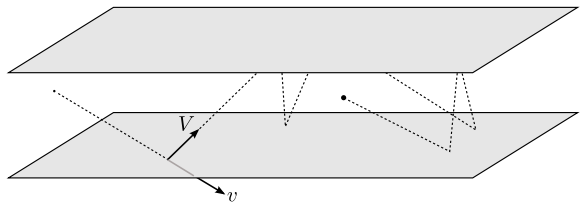}\ \ 
\caption{\small   An idealized molecular flight between two solid plates, as
an example of a multi-scattering experiment. We refer to $v$ and $V$, respectively,  as the pre- and post-scattering
velocities at  a  collision event, and  regard $V$ as a random function of $v$, as explained in the text.}
\label{multiscattering}
\end{center}
\end{figure} 

An example of a very simple interaction mechanism (in dimension $2$)
is shown in Figure \ref{movingwall}. That figure can be thought to represent 
a choice of wall ``microstructure.''  In addition to a choice of mechanical system representing the wall 
microstructure, the specification of a scattering map requires fixing a statistical
kinetic state of  this microstructure  prior to each collision. 
For the example of Figure \ref{movingwall}, one possible specification may be as follows:
(1) the precise position on  the horizontal axis (the dashed line of the figure) where the molecule enters
the zone of interaction is random, uniformly distributed over the period of the periodic surface contour;
(2) at the same time that the molecule crosses the dashed line (which arbitrarily 
sets the boundary of the interaction zone),  the position and velocity of the up-and-down moving wall 
are chosen randomly  from prescribed probability distributions. The most natural are the uniform distribution (over a
small interval)
for the position,
and a one-dimensional  normal distribution for the 
velocity, with mean zero and constant variance. (The variance specifies the wall temperature, as will be seen.)
In fact, one  general assumption of the main theorems essentially amounts to   the constituent masses of the wall sub-system 
having velocities which are  normally distributed and in a state of equilibrium (specifically, energy equipartition is assumed).
In this respect, a random-mechanical  model of ``heat bath'' is explicitly  given.  Once the
random pre-collision conditions are set, the mechanical system describing the interaction evolves
deterministically to produce $V$. Note that a single collision event may consist of several ``billiard collisions''
at the ``microscopic level.''
 
Having specified a scattering map  (by the choices of a mechanical system
and the  constant  pre-collision  statistical state of the wall), a random dynamical system on $\mathbb{H}$ is
defined, which can then be  studied from the perspective of the theory of Markov chains on general state spaces (\cite{meyn}). 

Clearly, one can equally well  envision a multiple scattering set-up similar to the one depicted in Figure \ref{multiscattering}
but inside a cylindrical channel or a spherical container rather than two parallel  plates; or, more generally,
inside a solid container of irregular shape, in which case a ``random change of frames''
operator must be composed with the scattering operator to account for
the changing orientation of the inner surface of the container at different collision points. (See \cite{feres}. This is not needed
in the case of plates, cylinders, and spheres.) We like to think of this general set-up as 
defining a {\em random billiard system}, an idea that is nicely illustrated by \cite{comets,comets2}, for example.

We are particularly interested in situations that exhibit weak scattering, in the sense that
the probability distribution of $V$ is concentrated near $v$ or, what amounts to the same thing,
the scattering is nearly specular.  Our systems will typically depend on a parameter $h$
that indicates the strength of the scattering, and we are mainly concerned with the limit of the  velocity (Markov) process
as $h$ approaches $0$.  This will lead to  novel types of  diffusion processes canonically associated with
the underlying   mechanical systems. 
We  call $h$ the {\em flatness} parameter for reasons that will soon become obvious.

For the systems of {billiard type} considered  here (introduced in Section \ref{billiardtypesection}),
the   essential information concerning their    mechanical and probabilistic  definition is contained
in two linear maps: $C$ and $\Lambda$ on $\mathbb{R}^{m+k}$, where  
$k$ is the number of  ``hidden'' independent  variables (whose statistical states are prescribed by the model)
and $m$ is the number  of   ``observed'' variables (say, the $3$ velocity coordinates of the molecule
in the situation of Figure \ref{multiscattering}). This $m$  is also the general dimension of $\mathbb{H}$. 
Vectors in  $\mathbb{H}$ will be written $v=(v_1, \dots, v_m)$. The maps $C$ and $\Lambda$ are
non-negative definite and Hermitian; $C$ is  a covariance matrix for the hidden velocities
and, by the equipartition assumption, it is  a scalar multiple of an orthogonal projection, while $\Lambda$ contains
(in the limit $h\rightarrow 0$) information about the system geometry and mass distribution.

The first observation (which is studied in much greater generality in \cite{scott})
  is that the resulting Markov chain on velocity space $\mathbb{H}$ 
  has   canonical stationary distributions given by what we refer to as the {\em post-collision} Maxwell-Boltzmann
  distribution of velocities. (The term ``post-collision'' is used to distinguish it from
  the more commonly known distribution of velocities sampled at random times,  not necessarily on the wall surface.)
This velocity distribution has the form
\begin{equation}\label{stationaryintroduction} d\mu(v)=|v_m| \exp\left(-\frac12|v|^2/\sigma^2\right)\, dV(v)\end{equation}
where $\sigma^2=\text{Tr}(C\Lambda)/\text{Tr}(\Lambda^\curlywedge)$,
$\Lambda^\curlywedge$ is the restriction of $\Lambda$ to the subspace of ``hidden velocities,''
and $dV$ denotes Euclidean   volume element.
 The 
     scattering map can be represented as a (very generally) self-adjoint operator on $L^2(\mathbb{H},\mu)$
     of norm $1$, which we indicate by $P_h$, where $h$ is the flatness parameter.
We denote the density of $\mu$ by $\varrho:=d\mu/dV.$  The term $|v_m|$ in $\varrho$ equals  the
speed times the cosine of the angle between $v$ and the normal to the scattering surface; this
cosine factor is often referred to in the applied literature as the {\em Knudsen cosine law.}  (\cite{celestini})

Let $C_0^\infty(\mathbb{H})$ denote the space of compactly supported smooth functions on the half-space.
A first order differential operator can be defined on this space using $C$ and $\Lambda$ as follows:
\begin{equation}\label{firstorderoperator}
\left(\mathcal{D}\Phi\right)(v):=\sqrt{2}\left[ \Lambda^{1/2}\left(v_m\, \text{grad}_v\, \Phi- \Phi_m(v) v\right) +\text{Tr}\left(C\Lambda\right)^{1/2}\Phi_m(v) e\right]
\end{equation}
where $e$ is the coordinate vector $(0, \dots, 0,1)$ and the subindex $m$ in $\Phi_m$ indicates 
partial derivative with respect to $v_m$. 
We  now define a second order differential operator on 
$C_0^\infty(\mathbb{H})$   by
$$ \mathcal{L}:=- \mathcal{D}^*\mathcal{D},$$
where $\mathcal{D}^*$ indicates the adjoint of $\mathcal{D}$ with respect to the natural
inner product on the  pre-Hilbert space of  smooth, compactly supported square integrable vector fields
with  the Maxwell-Boltzmann  measure $\mu$.  We refer to $\mathcal{L}$ as
the {\em MB-Laplacian} of the mechanical-probabilistic model, $\mathcal{D}$ as the {\em MB-gradient}, and $-\mathcal{D}^*$
 the {\em MB-divergence.}  

The central result of the paper is that, as $h$ approaches $0$ (and under commonly satisfied  further conditions
to be spelled out later), the Markov chain process with transition probabilities operator $P_h$ converges
to a diffusion process on $\mathbb{H}$ whose  infinitesimal generator  is the MB-Laplacian $\mathcal{L}$. 
The resulting   process, which we call {\em MB-diffusion}, is illustrated in a number of concrete examples in the paper.

\vspace{0.1in}
\begin{figure}[htbp]
\begin{center}
%\epsfile{file=bundle.eps,scale=0.8}
\includegraphics[width=3in]{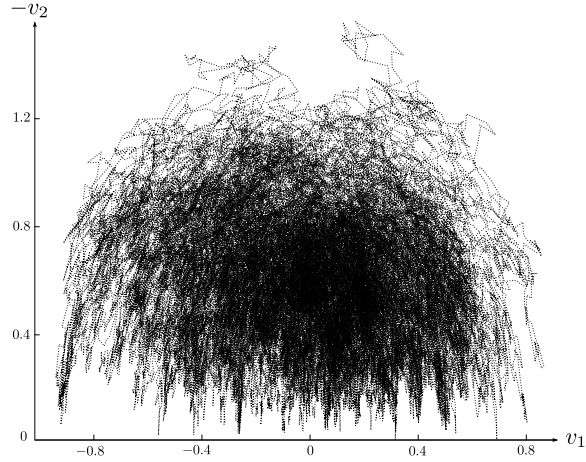}\ \ 
\caption{\small  Sample trajectory of the MB-diffusion process whose infinitesimal generator is given by
$(\mathcal{L}\Phi)(v)=-2v_1\Phi_1 +\left[-4v_2+\left(1+v_1^2\right)/v_2\right]\Phi_2+v_2^2 \Phi_{11}-2v_1v_2\Phi_{12}+(1+v_1)\Phi_{22}$
obtained by simple Euler approximation. We have used a time interval of length $50$,
initial condition $(0,-1)$ and number of steps $50000$. The parameters chosen here are not
related to those of Figure \ref{example3traj},  so the axis scales are not comparable.}
\label{example3SDE}
\end{center}
\end{figure}

The MB-diffusion can be expressed as an It\^o  stochastic differential equation
$$dV_t= Z(V_t)\, dt + b(V_t)\, dB_t,$$ where $B_t$ is   $m$-dimensional Brownian motion (restricted
to $\mathbb{H}$), $Z(v)$ is the vector field
$$Z(v):=-2\Lambda v +({\varrho_m}/{\varrho}) \left[\langle \Lambda v,v\rangle +\text{Tr}(C\Lambda)\right]e$$
and $b(v)$ is the linear map
$$ b(v)u:=v_m \Lambda^{1/2}u-\langle \Lambda^{1/2}v,u\rangle e_m +\text{Tr}\left(C\Lambda\right)^{1/2}u_m e. $$
A  sample path of an MB-diffusion  in dimension $2$ is shown in Figure \ref{example3SDE}.

It is interesting to note that, in dimension $1$, the operator $\mathcal{L}$ reduces to 
the (up to a constant) Laguerre differential operator (in this case on functions defined on the interval  $(-\infty, 0)$):
$$\frac{1}{2\lambda\sigma^2}(\mathcal{L}\Phi)(v)=\frac1\varrho \frac{d}{dv}\left(\varrho \frac{d\Phi}{dv}\right),$$
where $\varrho(v)=\sigma^{-2}v\exp\left(-v^2/2\sigma^2\right)$ is the Maxwell-Boltzmann density and $\lambda$
is the scalar equal to the (in this case $1$-by-$1$) matrix denoted above by  $\Lambda$.

\vspace{0.1in}
\begin{figure}[htbp]
\begin{center}
%\epsfile{file=bundle.eps,scale=0.8}
\includegraphics[width=2.5in]{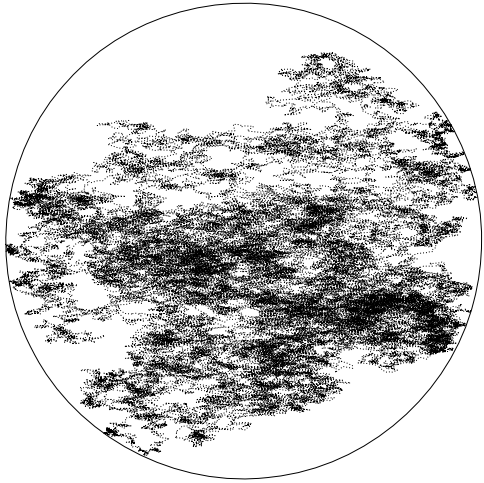}\ \ 
\caption{\small  Sample path for the Legendre diffusion in dimension $2$ with eigenvalues (of $\Lambda$)
$\lambda_1=2.5$ and $\lambda_2=1$. The starting point is $(0,0)$, the time length is $5$, and the
number of steps is $50000$.  Note that diffusion is faster along the horizontal axis as $\lambda_1>\lambda_2$.
The stationary distribution for this process is the normalized Lebesgue measure; therefore,  long trajectories    fill the disc evenly.}
\label{legendre}
\end{center}
\end{figure}

It has been implicitly assumed above that the number of independent ``hidden'' velocity components  $k$ is positive.
The case $k=0$ is somewhat different and has special interest. Now, particle speed does not change
after scattering (we refer to this case as {\em random elastic} scattering) and both the Markov chain process and
the diffusion approximation can be restricted to a unit hemisphere in $\mathbb{H}$. Alternatively,
by orthogonal projection from the hemisphere to the unit open ball in dimension $m-1$, 
we can consider these processes taking place on the ball. Now $\mathcal{L}$ is  
 completely determined by $\Lambda$ and 
has the form
$$ (\mathcal{L}\Phi)(v) = 2\sum_{i=1}^n \lambda_i  \left((1-|v|^2) \Phi_i\right)_i$$
where the $\lambda_i$ are the eigenvalues of $\Lambda$ and  $\Phi$ is a compactly supported smooth function on the unit ball. 
(We have chosen coordinates adapted to the eigenvectors of $\Lambda$.)

The associated diffusion
process, written as an It\^o stochastic differential equation, has the form 
$$dV_t= -4\Lambda V_t\, dt + \left[2\left(1-|V_t|^2\right) \Lambda\right]^{1/2}dB_t,$$ where $B_t$ is $(m-1)$-dimensional Brownian motion
restricted to the unit ball. The operator $\mathcal{L}$ in this case naturally generalizes the
standard Legendre differential operator on the unit interval $(-1,1)$; we call the
stochastic process on the higher dimensional balls   {\em Legendre diffusions}.
The stationary measure turns out to be standard  Lebesgue measure on the ball, so Legendre diffusions
have the interesting property that sample paths fill the ball uniformly with probability $1$.
A sample path of a Legendre diffusion process is illustrated in Figure \ref{legendre}.
One can also think of the Legendre diffusion as a special case of the MB-diffusion
in the sense explained in  Proposition \ref{reductionLegendre} and in the remarks immediately 
after this proposition.

The relationship between the scattering operators $P$ and the above  differential operators of Sturm-Liouville
type suggests that one should be able  fruitfully to  investigate the spectral theory of $P$ based on an analysis of $\mathcal{L}$
and a  spectral perturbation approach.  A very simple observation in this regard is indicated in
\cite{fz}, while  \cite{fz2} discusses the spectral gap of $P$ (which can often be shown to be a Hilbert-Schmidt operator)
in very special cases.  We hope to turn to a more detailed analysis  of the spectrum of $P$ in a future study.

 \section{Mechanical systems of billiard type}\label{billiardtypesection}
 We introduce in this section the main  definitions and basic facts  concerning  
 classical  mechanical systems of {\em billiard type} and their derived random systems.
Component masses  of a given     mechanical model interact via elastic scattering that admit
 a {\em billiard representation}.
    Particular attention is given to {\em weak scattering},  in which reflection in this billiard representation
     is nearly specular.
     A sequence of random scattering events comprises a Markov chain on velocity space whose transition probabilities operator,
     in the case of weak scattering,  is close to the identity.

 \subsection{Deterministic scattering events}
 The reader may like to keep in view the examples of Figures \ref{convex} and \ref{movingwall} while
 reading the below definitions.
For our purposes, a system of {\em billiard type} is a mechanical system   defined  by the geodesic motion of
a point particle in  a Riemannian manifold of dimension $n$ with piecewise smooth boundary.  
Upon hitting the boundary, trajectories  reflect back into the
interior of the manifold  according to
ordinary specular reflection and continue along a geodesic path.  
Except for passing references to more general
situations, the configuration manifolds  of the systems considered in  this paper
are  {\em Euclidean}. 
More specifically, we   consider
  $(n+1)$-dimensional  submanifolds  of $ \mathbb{T}^n\times\mathbb{R}$ with boundary,  for some $n$,
with a  Riemannian metric which  will have constant coefficients with respect to the standard coordinate system. 
The boundary, assumed to be the graph of a piecewise smooth  function,  and the metric coefficients are the distinguishing features of each model.
The periodicity implied by the torus factor   is a more restrictive condition than really needed, 
but these  manifolds are   a natural first step  and describe a variety of situations of special interest.
Thus  the configuration manifold $M$ is   assumed to have the form
 $$M=\{(x,x_{n+1})\in \mathbb{T}^n\times \mathbb{R}: x_{n+1}\geq F(x)\}$$
where $F$ is a piecewise smooth function. 

The Riemannian metric on $M$ is specified  by the kinetic energy quadratic form, 
 which depends on the distribution of masses in the system. The Euclidean condition means, in effect,
that the kinetic energy form becomes, after a linear coordinate change,   the standard dot-product norm restricted to (the tangent bundle of) $M$, while
the torus component of $M$ has the form
 $\mathbb{T}^n=\Pi_{i=1}^n\left(\mathbb{R}/a_i \mathbb{Z}\right)$ for positive constants $a_1, \dots, a_n$.

A {\em scattering event}  is defined   as follows.
Let $c$ be an arbitrary  constant satisfying  $F(x)<c$ for all $x\in \mathbb{T}^n$. The
submanifold $x_{n+1}=c$ will be called the {\em reference plane}. We identify the tangent space to $M$
at any point  on the reference plane with $\mathbb{R}^{n+1}$ and denote by $\mathbb{H}_-^{n+1}$ the lower-half space
in $\mathbb{R}^{n+1} $, which consists of tangent vectors whose $(n+1)$st coordinate is negative. 
\begin{definition}[Deterministic scattering event]\label{definitionevent}
A scattering event
is an iteration of the  correspondence $(x,v)\mapsto (x',V)$, where $x, x'$ lie on the reference plane  and $(x', V)$ is
the end state of a billiard trajectory that begins at $x$ with velocity $v$ and ends at $x'$ with velocity $V$.
By reflecting $V$ on the reference plane, we may
when convenient
 regard both $v$ and $V$ as vectors in $\mathbb{H}_-^{n+1}$.
\end{definition}

Notice that the map describing a scattering event is indeed well defined, at least for almost all $(x,v)$, by Poincar\'e recurrence.
If  $|\text{grad}_xF|$ is uniformly small over $x\in \mathbb{T}^n$, a condition that is assumed in the main theorems,
trajectories cannot get trapped.

The iteration of the scattering event map  introduced in Definition \ref{definitionevent},
as well as its   associated random maps described in Subsection \ref{definitionsscatter},
acquires greater significance in the context of random billiard systems as   in 
 \cite{scott}, but the various concrete examples given later in this paper (the simplest of
 which appears in  Subsection \ref{flatfloor})   should provide enough
 motivation.

\subsection{Random systems and weak scattering}\label{definitionsscatter}
The random scattering set-up defined here  is a special case of the one considered in \cite{scott}.
Briefly, the main idea is that some of the variables involved in a  deterministic scattering 
event, as defined above, are  taken to be random. The scattering map then becomes a random function 
of the initial state of the system. 
The resulting random system can model a variety of physical situations; we refer to \cite{scott} for
more details on the physical interpretation.   

The notation  $\mathcal{P}(X)$ will be
used below to designate the space of probability measures on a measurable space $X$.
We start with a deterministic scattering system with configuration manifold $M\subset \mathbb{T}^n\times\mathbb{R}$ and boundary function $F:\mathbb{T}^n\rightarrow \mathbb{R}$. 
Recall that $M$ is  defined by
the inequality $x_{n+1}\geq F(x)$, $x\in \mathbb{T}^n$. 
The deterministic scattering map is  then  $(x,v)\mapsto (x',V)$,
where $x,x'$ lie on the reference plane $x_{n+1}=c$ (recall that $c$ is an essentially arbitrary value that
specifies the reference plane);  $v$ and $V$ lie in  the lower-half space $\mathbb{H}_-^{n+1}$,
and $V$ is the  reflection on the reference plane of the   velocity of the billiard trajectory with
initial state $(x,v)$ at the moment the trajectory returns to the reference plane. A scattering event
can  consist of several billiard collisions.

Choose $c'$ such that $\sup_{x}|F(x)|\leq c'<c$ and define
$$M_{c'}:=\{(x,x_{n+1})\in M:x_{n+1}>c'\}=\mathbb{T}^n\times (c', \infty).$$
Let $k\leq n$ be a non-negative integer and write $M_{c'}=\mathbb{T}^k\times \mathbb{T}^{n-k}\times (c',\infty)$.
Accordingly, decompose the tangent space to $M$ at any point on the reference plane 
as $\mathbb{H}^{n+1}_-=\mathbb{R}^k \times \mathbb{H}_-^{n-k+1}$.
Fix a probability measure $\mu$ on $\mathbb{R}^k$ and set  $m=n-k+1$.
By a {\em  random initial state with observable component} $v\in \mathbb{H}_-^{m}$ we mean a state of the
form $(x, c, w, v)$, where $x\in \mathbb{T}^n$ is a uniformly distributed  random variable, $c$ is the value defining the
reference plane, and  $w$ is a random variable taking values in $\mathbb{R}^k$ with probability measure $\mu$.
To this  random initial state  we can associate a probability measure $\nu_v\in \mathcal{P}(\mathbb{H}^m_-)$ as follows:
Consider the 
 trajectory of the system of billiard type having
random initial state  $(x, c, w, v)$,  and
let $V$ be the component in  $\mathbb{H}_-^{m}$ of the final velocity of this trajectory, reflected back into $\mathbb{H}_-^m$,
at the moment   the trajectory   returns to the reference plane. Then $V$ is
a random variable and $\nu_v$ is by definition  its probability measure. We refer to $\nu_v$ as the {\em return probability distribution}
associated to the random initial state having observable component $v$.

\begin{definition}[Random scattering event]
Let $\mu$ be a probability measure on $\mathbb{R}^k$ and give $\mathbb{T}^n$ the uniform  probability measure, denoted $\lambda$.
Then the {\em random scattering event} associated to the system of billiard type and these fixed measures
 is defined by  the map $$v\in \mathbb{H}^{m}_-\rightarrow \nu_v\in \mathcal{P}(\mathbb{H}^m_-),$$
 where $\nu_v$ is the return probability associated to the  random initial state with observable component $v$.
\end{definition}

The probability measure $\mu$ 
on $\mathbb{R}^k$  typically will be assumed to have zero mean,  non-singular   covariance matrix
  of finite norm, and finite moments of order $3$, when not assumed more concretely to be Gaussian. 
   The uniform distribution on
 $\mathbb{T}^n$ is, by definition, the unique translation invariant   probability measure.

\begin{definition}[Scattering operator $P$]
Let $C_0(\mathbb{H}_-^{m})$ denote the space of compactly supported continuous functions on 
the lower half-space. For any given $\Phi\in C_0(\mathbb{H}_-^{m})$, define
$$(P\Phi)(v):=\int_{\mathbb{R}^k}\int_{\mathbb{T}^{n}}\Phi(V(x,c, v, w))\, d\lambda(x)\, d\mu(w). $$
We call $P$ the {\em scattering operator} of the system for a random initial state specified by $\mu$
and the uniform distribution on the torus.
\end{definition}

Operators similar to our $P$ naturally arise in kinetic  theory of gases 
and are used to specify boundary conditions for the Boltzmann equation. 
See, e.g., \cite{cercignani,harris}. Typically, the models of gas surface interaction
used in the Boltzmann equation literature are phenomenological, such as
the Maxwell model (\cite{cercignani}, Equation 1.10.20), and are not derived from
explicit mechanical interaction models as we are interested in doing here.

From the definitions it follows that  $P$ and $\nu_v$ are related by
$$ (P\Phi)(v)=E_v[\Phi(V)]=\int_{\mathbb{H}_-^{m}} \Phi(u)\, d\nu_v(u)$$
where the expression in the middle denotes the expectation of the random variable $\Phi(V)$
given the initial condition $v$. 

Based on the  examples  given  throughout the paper, we can expect   $P$  and  $v\mapsto \nu_v$  generally to have 
good  measurability properties, due to the deterministic map from which the random process is defined 
  being typically piecewise smooth.
In  our general theorems it will be implicitly  assumed  that $v\mapsto \nu_v(A)$  is Borel measurable 
for all Borel measurable subsets $A$ of $\mathbb{H}_-^m$. Billiard maps are typically not continuous;
see \cite{chernov} for basic facts on billiard dynamics (in dimension $2$).
 
The following additional assumption turns out to be  convenient and not too restrictive.
\begin{definition}[Symmetric $M$, $F$]\label{symmetric}
The configuration manifold $M$ of  a random scattering process or, equivalently, the  function $F$ defining it, will
be called    {\em symmetric} if   $F(o+u)=F(o-u)$ for all $u\in \mathbb{R}^n$  and some choice
of origin $o$ in $\mathbb{T}^n$.
 \end{definition}

\subsection{Example: collision of a rigid body and flat floor}\label{flatfloor}

A simple example will help to  clarify and motivate some of the above  definitions.  (More representative
examples will be introduced later.)
Consider the $2$-dimensional system of Figure \ref{convex}. 
It consists of a rigid body in dimension $2$ of constant density 
and mass $m$  that
moves in the half-plane set by a hard straight floor. There are no potentials (e.g., gravity). 
The body and floor surfaces are assumed to be physically smooth, in the sense that 
there is no change  in the component of the linear  momentum  tangential to the floor after a collision. 
The motion of the center of mass can then   be restricted to the dashed line of Figure \ref{convex}
due to conservation of the horizontal component of the linear momentum.

 \vspace{0.1in}
\begin{figure}[htbp]
\begin{center}
%\epsfile{file=bundle.eps,scale=0.8}
\includegraphics[width=2.0in]{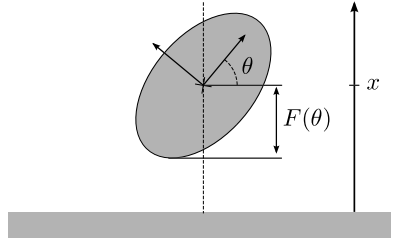}\ \ 
\caption{\small  Collision of a rigid body and a flat floor in dimension $2$. The shape of the body is
encoded in the function $F(\theta)$. For a disc, $F$ is constant. Some of the main results of the paper apply to shapes for
which $h:=\sup_{\theta}\left(F'(\theta)\right)^2$  is small.   }
\label{convex}
\end{center}
\end{figure}

Figure \ref{billconvex} shows the  description of the same example explicitly as a system of billiard type.
 Let $B$ represent the body at a fixed position, with its  center of mass at
the origin. Define the second moment of the
position vector by   $l^2:=\text{Area}^{-1}\int_B |b|^2 \, dA(b)$, where  $A$ is the
area measure. Set 
coordinates $x_1=\theta$ and $x_2=x/l$, where $\theta$ is the angle of rotation and $x$ is the height of the center of mass 
of the body
at a given configuration   in $\mathbb{R}^2$.  
Then the configuration manifold of the system
is the  region
$M=\{(x_1, x_2)\in \mathbb{T}\times\mathbb{R}: F(x_1)\leq l x_2\},$ equipped with the kinetic energy metric $K=\kappa\left(\dot{x}_1^2+\dot{x}_2^2\right)$, where $\kappa$ is a positive constant.
We model the collision between the body and the floor by a linear map $C:T_xM\rightarrow T_x M$, where $x$ is
a boundary (collision) point of $M$. Under the assumption of energy conservation and time reversibility, 
$C$ is an orthogonal involution; the assumption of physically smooth contact is interpreted as $Cu=u$ for every
 nonzero vector $u$  tangent to the boundary at  $x$. As   $C$ cannot be the identity map, it must be 
standard  Euclidean  reflection, whence  the system is of billiard type.

   The $3$-dimensional version of this example is similarly described, the function $F$  now being  defined on the special orthogonal group
 $\text{SO}(3)$.   The
 kinetic energy   metric on $M\subset \text{SO}(3)\times \mathbb{R}$ is no longer Euclidean and  naturally involves the body's moment of inertia.

 \vspace{0.1in}
\begin{figure}[htbp]
\begin{center}
%\epsfile{file=bundle.eps,scale=0.8}
\includegraphics[width=1.5in]{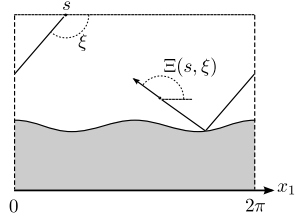}\ \ 
\caption{\small   The billiard representation  associated to the mechanical system of Figure \ref{convex}.  Here $M\subset \mathbb{T}\times \mathbb{R}$. (Dashed lines indicate periodic conditions.) The reference plane is indicated by the top horizontal   line 
and the wavy ground is the graph of $F$. 
We define  the pre-scattering    angle $\xi$ and   initial position $s\in [0,2\pi]$. The  
outgoing angle   is $\Xi$.}
\label{billconvex}
\end{center}
\end{figure}

One way in which this deterministic system can be turned into 
an example of a  random system   is by regarding the initial angle $\theta$, at
the moment the center of mass of the body crosses a reference plane, to be random with the uniform
distribution over the interval $[0,2\pi]$.
In other words, 
suppose that the exact orientation of the
body in space at a  given moment prior to collision is completely  unknown. As the magnitude of the velocity of the billiard particle (that is, 
of the moving point particle of Figure \ref{billconvex})  is invariant throughout the process due to energy conservation, we may consider the return probability
as being supported  on the half-circle in $\mathbb{H}^2_-$, which we identify with the interval of angles $[0,\pi]$.

Thus the   probability distribution $\nu_\xi$ of the return angle $\Xi$ given $\xi$ is   the measure:
$$U\mapsto \nu_\xi(U)=\frac{1}{2\pi}\int_0^{2\pi} \mathbbm{1}_U\left(\Xi(s,\xi)\right)\, ds $$
where $U$ is a measurable subset of the interval $[0,\pi]$ (a set of scattered angles) and $\mathbbm{1}_U$
is the indicator function of $U$. 
Similarly, given a continuous function $\Phi$ on $[0,\pi]$, 
$$(P\Phi)(\xi)=\frac{1}{2\pi}\int_0^{2\pi}\Phi\left(\Xi(s,\xi)\right)\, ds=\int_0^\pi \Phi(\Xi)\, \nu_\xi(\Xi).$$
The probability distributions of the velocity  of the center of mass and the angular
velocity (expressed in terms of $\dot{x}_2$ and $\dot{x}_1$, respectively)  are obtained by taking the push-forward of $\nu_\xi$ under
the maps $u\mapsto |\xi| \cos(u)$ and $u\mapsto |\xi| \sin(u)$, respectively.
Note that $\nu_\xi$ approaches weakly the delta measure $\delta_\xi$ supported on $\xi$ when
 the body  becomes more and more round (hence the reflecting line in Figure \ref{billconvex} becomes more
 and more straight).    
In this case,   $P$ approaches the identity operator.

\subsection{Further notations}\label{further}
All the   examples  of deterministic systems given  in this paper  can be turned into random scattering systems in various 
ways. 
The most natural choices of random variables fall within the scope of the following discussion,
in which the definition of a random scattering event is restated in a more convenient  form.
Let the tangent space to $M$ at any point $(x,c)$ of the reference plane 
decompose in the following two  different ways: $$\mathbb{H}_-^{n+1}=\mathbb{R}^k\times \mathbb{H}^{m}=\mathbb{R}^{n}\times (-\infty,0).$$
where $m=n-k+1$. Accordingly, any given $\xi\in T_xM$ has components  relative to these two decompositions
defined by
$$\xi=(\xi^\curlywedge, \xi^\curlyvee)=(\overline{\xi}, \xi_{n+1}). $$
Let $\{e_1, \dots, e_{n+1}\}$ be the standard basis of $\mathbb{R}^{n+1}$  and $e:=e_{n+1}$ the last basis vector.
So $$\xi_{n+1}=\langle \xi, e\rangle, \ \ \xi^\curlywedge= \sum_{i=1}^k \langle \xi, e_i\rangle e_i,\ \ 
\xi^\curlyvee=\xi-\xi^{\curlywedge}$$
where the inner product represented by the angle brackets is the standard dot product.
The component $\xi^\curlyvee$ of the final velocity of the  billiard   trajectory is the quantity 
of interest produced  by the scattering event. The component $\xi^\curlywedge$  of the initial velocity 
is assumed random with a probability distribution $\mu$.

 \vspace{0.1in}
\begin{figure}[htbp]
\begin{center}
%\epsfile{file=bundle.eps,scale=0.8}
\includegraphics[width=3.5in]{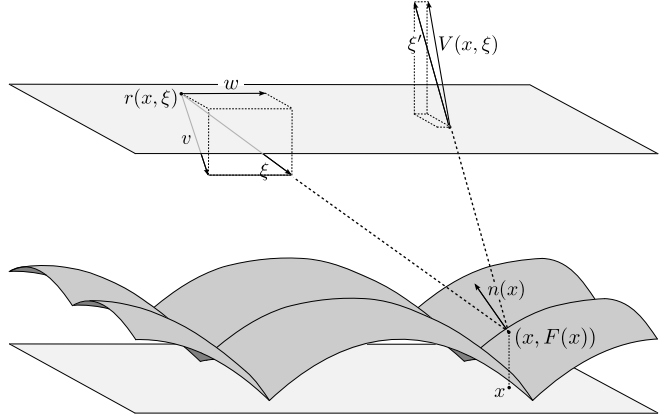}\ \ 
\caption{\small  The random collision process.  The initial position $r\in \mathbb{T}^{n}\times\{c\}$ on the reference plane is chosen randomly  with the uniform
distribution. The outgoing velocity in $TM$ is $\xi'$, and the vertical projection $V:=\xi^\curlyvee$
is the outgoing observed velocity. We thus obtain a random
map $v\mapsto V$.
It will turn out to be convenient to express both  $r$  and $V$ in terms of $\xi$ and the independent variable $x\in \mathbb{T}^{n}$, as indicated in the figure.
 }
\label{definitions3}
\end{center}
\end{figure}

We can now express the random scattering map $v\mapsto V$ by the following algorithm, which is illustrated in Figure
\ref{definitions3}.
\begin{definition}[Random scattering algorithm]\label{steps} In the notation introduced above, the random scattering map is defined by
the following steps:
\begin{enumerate}
\item[i.] Start with $v\in \mathbb{H}_-^{m}$;
\item[ii.] Choose a random $w\in \mathbb{R}^k$ with the probability distribution $\mu$ and form $\xi:= w+v$;
\item[iii.] Choose a random $r=(\overline{r},c)$, where  $\overline{r}\in \mathbb{T}^n$ is uniformly distributed, and let $(r, \xi)$  be the initial state for the billiard trajectory;
\item[iv.] Let $\xi'\in \mathbb{H}_-^{n+1}$ be the velocity of the billiard trajectory at  the time of its  return to $r_{n+1}=c$;
\item[v.] Set $V:=(\xi')^\curlyvee\in \mathbb{H}_-^{m}$.
\end{enumerate}
\end{definition}

For billiard surfaces (i.e., the graph of $F$)
that are relatively flat, the typical collision process comprises a single collision. 
It makes sense in this case  to introduce the independent variable $x$ as indicated in Figure \ref{definitions3},
and allow both $\overline{r}$ and $\xi'$ to be functions of $x$ and $\xi$.
For   $x\in \mathbb{T}^{n}$, let $n(x)$ be the unit normal vector to the graph of 
$F$ at  $(x, F(x))$. 
Note that 
$$n(x)=\frac{e - \text{grad}_xF}{\sqrt{1+\|\text{grad}_xF\|^2}}.$$

As a geometric measure of the strength of scattering we introduce the following parameter.

\begin{definition}[Flatness parameter]\label{flatness}
The quantity $h:=\sup_{x\in \mathbb{T}^n}|\text{\em grad}_xF|^2$ will be referred to as the {\em flatness parameter} of the
system defined by $F$.  
\end{definition}

\section{Stationary measures and general  properties of $P$}
 The basic properties of $P$ are described in this section.
These are mostly special cases of results from \cite{scott}, which we add here for easy reference.
Proofs  are much simpler in our present setting and are sketched  here.
\subsection{Stationary measures}\label{invariantstandard}
It will be assumed in much of  the rest of the paper that
 the probability distribution $\mu$ for the velocity component $w$ in $\mathbb{R}^k$, when $k>0$, is Gaussian:
$$d\mu(w)=\frac{e^{-\frac12|w|^2/\sigma^2}}{\left(2\pi\sigma^2\right)^{k/2}} dV(w)$$
occasionally referring to $\sigma^2$ as the {\em  temperature} (of the ``hidden state'').
Let  $\lambda$ be the translation invariant probability measure on $\mathbb{T}^n$. The 
standard  volume element in open subsets of Euclidean space will be   written   $dV$, or $dV^k$
if we wish to be explicit about the dimension.

Recall that the deterministic scattering map $T$ is defined as the return billiard flow map on the phase space
restricted to the reference plane, that is,
$\mathbb{T}^n\times \mathbb{H}_-^{n+1}=\left(\mathbb{T}^n\times \mathbb{R}^k\right)\times \mathbb{H}_-^m$,
where we are factoring out the observable velocity component $\mathbb{H}^m$ from
the hidden  component at temperature $\sigma^2$, which    is  given the  probability distribution  $\nu:=\lambda\otimes  \mu$
in the sense described in the previous section. 
Denote by $\pi$   the natural projection    $\pi: \mathbb{T}^n\times \mathbb{R}^k\times  \mathbb{H}_-^m\rightarrow \mathbb{H}_-^m$

When $k=0$ (no  velocity components among the hidden variables), the scattering interaction does not change the magnitude of the
velocity in $\mathbb{H}_-^m$; thus one may restrict the state to the unit hemisphere $S_-^{m-1}$ in $\mathbb{H}^m_-$.
Let $d\omega(u)$ represent the Euclidean volume element (measure) on the hemisphere at the unit vector $u$. When necessary we indicate the
dimension of the unit hemisphere as
$d\omega^{m-1}$.  Observe that $$dV^m(v)=c |v|^{m-1}\, d\omega^{m-1}(v/|v|)\, d|v|,$$ where 
$c$ is $m$ times the ratio of the volume of the unit $m$-ball by the volume of the unit $(m-1)$-sphere.

The Markov operator $P$ naturally acts on probability measures on $\mathbb{H}_-^{m}$ 
as follows. With the notation $\eta(f):=\int f\, d\eta$,   the action of $P$ on $\eta$ is the measure $\eta P$ such that
  $(\eta P)(f)=\eta(Pf)$, for every compactly supported continuous  $f$. 
A probability measure $\eta$ on $\mathbb{H}_-^m$ is said to be {\em stationary} for a Markov operator with
state space $\mathbb{H}_-^m$ if $\eta P=\eta$.  

The action of $P$ on probability measures has the following convenient expression. 
Given any probability measure  $\eta$ on $\mathbb{H}_-^m$, we can form the probability measure
$\nu\otimes \eta$ on $(\mathbb{T}^n\times\mathbb{R}^k)\times\mathbb{H}_-^m$, then
act on this measure by the push-forward operation  $T_*$ under the return map, and finally project
the resulting  probability measure back to $\mathbb{H}_-^m$. 
(We recall that $T_*\zeta$, for a given measure $\zeta$, can be defined by its evaluation on continuous functions as
$(T_*\zeta)(f):=\zeta(f\circ T)$.) 
 The result is  $\eta P$.
\begin{lemma}\label{simplelemma}
The operation $\eta\mapsto \eta P$ for  $\eta\in \mathcal{P}\left(\mathbb{H}_-^m\right)$ can be
expressed as 
$$ \eta P= (\pi\circ T)_*( \nu\otimes \eta),$$
where $\nu=\lambda\otimes \mu$ is the fixed probability on the hidden variables space $\mathbb{T}^n\times \mathbb{R}^k$,
$T$ is the return map to the phase space restricted to  reference plane, identified with
$\mathbb{T}^n\times \mathbb{R}^k\times \mathbb{H}_-^m$,  and $\pi$ is the projection from this
phase space to $\mathbb{H}_-^m$.
\end{lemma}
\begin{proof}
The straightforward  proof amounts to interpreting the definition of $P$ given earlier in terms of the push-forward notation.
See \cite{scott} for more details. 
\end{proof}

\begin{proposition}\label{invariantproposition}
When $k=0$, the   measure $d\eta(v)= \langle v, e\rangle \, d\omega(v)$ defined on the unit hemisphere in $\mathbb{H}_-^m$
is stationary under $P$. Identifying the unit hemisphere with $D_1^{m-1}=\{x\in \mathbb{R}^{m-1}: |x|<1\}$
   under the linear projection $(x,x_m)\mapsto x$,
the stationary probability is, in this case, the normalized Lebesgue measure on $D_1^{m-1}$.
For $k>0$,
the   measure
$$ d\eta(v)= | \langle v,e\rangle| e^{-\frac12 |v|^2/\sigma^2}\, dV^m(v)$$
on $\mathbb{H}_-^m$   is stationary  under $P$. 
\end{proposition}
\begin{proof}
For a much  more general result see \cite{scott}. We briefly show here the second claim.
First note that the measure $d\zeta_0(\xi):= \langle\xi,e\rangle \, d\lambda(x)\, dV(\xi)$ on $\mathbb{T}^n\times \mathbb{H}_-^{n+1}$
is    $T$-invariant. (The term $\langle \xi, e\rangle$ contains the cosine factor that appears in the canonical invariant measure
of billiard systems in general dimension.)
The measure 
$d\zeta(\xi)= \exp(-\frac12|\xi|^2/\sigma^2)\, d\eta_0(\xi)$ is also $T$-invariant,
since  any function of $|\xi|$ is  invariant under the return map $T$. Now, 
consider the  
  decomposition $\xi=(x, w,v)\in \mathbb{T}^n\times \mathbb{R}^k\times  \mathbb{H}_-^m$, under which
 $\zeta$ splits  as
$\zeta=\nu\otimes \eta$,
where $$d\eta(v)= | \langle v,e\rangle| e^{-\frac12 |v|^2/\sigma^2}\, dV^m(v).$$
Here we have used: $\langle \xi, e\rangle=\langle v,e\rangle$ and the splitting of the exponential involving
$|\xi|^2=|w|^2+|v|^2$ as a product of exponentials in $w$ and $v$. Thus $\pi_*\zeta=\eta$,
where $\pi_*$ indicates the push-forward operation on  probability measures.
We now apply Lemma \ref{simplelemma}, noting that $T_*\zeta =\zeta$, to obtain
 $$\eta P= (\pi\circ T)_* (\nu\otimes \eta)=\pi_* T_* \zeta=\pi_*\zeta=\eta,$$
which is the claim.
\end{proof}

There is a significant literature  in both pure mathematics and physics/engineering  concerning  {\em random billiards},
in which ordinary specular billiard reflection is replaced with a random reflection. The typical
assumption is that the post-collision velocity distribution corresponds to the above Maxwellian distribution
or, more simply, to the Knudsen cosine law with constant speed. See, for example, \cite{celestini,comets,comets2}.

It is natural to regard $P$ as an operator on the Hilbert space $L^2(\mathbb{H}_-^m,\eta)$ of
square integrable functions on the observable factor with the stationary measure given in Proposition \ref{invariantproposition}.
The next (easy) proposition is proved in \cite{scott}.

\begin{proposition}
Suppose that the billiard system is symmetric, as defined in Definition \ref{symmetric},
and let $\eta$ be one of the stationary measures described in Proposition \ref{invariantproposition}. Then
 $P:L^2(\mathbb{H}_-^m,\eta)\rightarrow L^2(\mathbb{H}_-^m,\eta)$
is a self-adjoint operator of norm $1$.
\end{proposition}

\subsection{Example:  collision between particle and moving surface} \label{movingsurface}
The example discussed here is the simplest that exhibits most of the features of the general case.
Its  components  are a point mass $m_1$ and 
a wall that is allowed to move  up and down. See   Figure \ref{movingwall}. 
The wall surface, which could be
of any dimension $n\geq 0$,
has a periodic, piecewise smooth contour  and the up and down motion is restricted to an interval $[0,a_0/2]$.
In  the interior of
this interval the wall moves freely, bouncing  off   elastically at the heights $0$ and $a_0/2$. 
Collisions between the wall and mass $m_1$ are also elastic.

      \vspace{0.1in}
\begin{figure}[htbp]
\begin{center}
%\epsfile{file=bundle.eps,scale=0.8}
\includegraphics[width=3.5in]{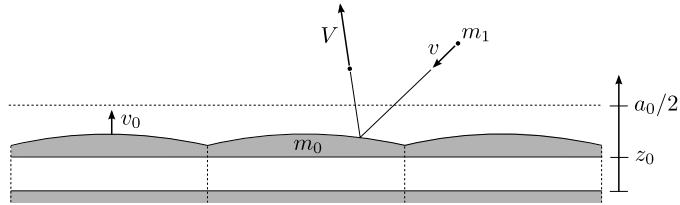}\ \ 
\caption{\small   Collision between the   surface of a moving rigid  wall   and a point mass $m_1$. 
The entire   wall is given a certain mass $m_0$; models with
more localized mass definition are also possible.}
\label{movingwall}
\end{center}
\end{figure}

Let $x$ denote the coordinate along the (horizontal) base of the wall and let $z_0$ be the height  at
which the  base stands at any given moment  relative to its lowest position.
The range of $z_0$ is assumed to be $[0,a_0/2]$. The  contour of the wall top surface, when  $z_0=0$,  is described by a periodic function $f(x)$
of period $a_1$.  Thus,  when the base is at height $z_0$, that contour  is the graph of
$x\mapsto f(x)+z_0$.  It is convenient to allow $z_0$ to vary over the symmetric interval 
$[-a_0/2, a_0/2]$ and set the wall surface function of $x, z_0$ as $G(z_0,x)=f(x)+|z_0|$, which can then be
extended periodically over $\mathbb{R}^2.$ The graph of $G$ so extended is, up to a rescaling of
the coordinates to be described shortly,  
the surface shown in   Figure \ref{example3}. 
Periodicity  of $G$ is expressed by  $G(x+ma_1,z_0+n a_0)=G(x,z_0)$ for integers
$m, n$.  Equivalently, we think of $G$ as a function on the $2$-torus.

 \vspace{0.1in}
\begin{figure}[htbp]
\begin{center}
%\epsfile{file=bundle.eps,scale=0.8}
\includegraphics[width=2.5in]{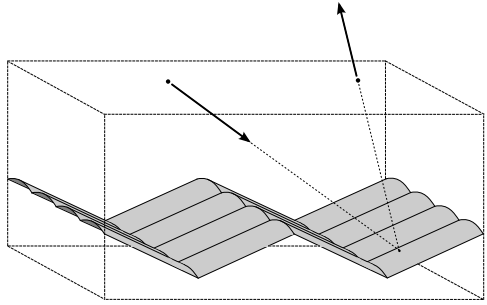}\ \ 
\caption{\small  
Billiard representation  for the moving wall example.}
\label{example3}
\end{center}
\end{figure}

The coordinates of the point mass $m_1$ are represented by 
  $(z_1,z_2)$  respectively along and perpendicular to the base line of the wall.
  Thus the state of the system
 at any  moment  is specified by $(z_0, z_1, z_2, v_0,  v_1, v_2)$, where $v_0$ is
 the velocity of $m_0$ and  $v=(v_1,v_2)$ is the velocity of $m_1$.

   An appropriate choice of coordinates makes  the kinetic energy metric explicitly Euclidean.  
Set $x=(x_0, x_1,   x_2)$, where 
 $x_0:=\sqrt{{m_0}/{m_1}}  {z_0}/{a_1}, \  x_1:= {z_1}/{a_1},\     x_{2}:={z_{2}}/{a_1}.
$
The above function $G$ in this new system becomes $$F(x_0,x_1)=a_1^{-1}  f\left(a_1x_1\right)+\sqrt{{m_1}/{m_0}}|x_0|.$$
Defining  $\tau:=\frac{a_0}{a_1}\sqrt{\frac{m_0}{m_1}}$, then  
$ F\left(x_0+m\tau, x_1+n\right)=F(x_0,x_1)$ for integers $m, n$.
The  configuration manifold of the particle-movable wall
system can now  be written in terms of $F$  as
$$ M=\{x\in \mathbb{R}^{3}:x_{2}\geq  F(x_0,x_1)\}.$$ 
The kinetic energy of the system then becomes
 $ K(x,\dot{x})=K_0  \|\dot{x}\|^2, $
where $K_0=m_1a_1^2/2$. 
Under the assumption that the wall surface is physically smooth, we obtain again a system of billiard type in $M$,
as depicted in Figure \ref{example3}.

A random billiard scattering process based on the above set up can now be defined as follows. The observable state space is the set $\mathbb{H}_-^2$
of approaching velocities,     consisting of the vectors  $u_1e_1+u_2 e_2$, $u_2<0$. 
The part of the  phase space of the deterministic process  on which the return map $T$ is defined is  
 $\mathbb{T}^2\times \{0\}\times \mathbb{R}\times \mathbb{H}_-^2$. Observe that  $\mathbb{T}^2=(\mathbb{R}/\tau\mathbb{Z})\times (\mathbb{R}/\mathbb{Z})$
has coordinate functions  $(x_0,x_1)$, and  the reference plane, with equation $x_2=0$, is  identified with $\mathbb{T}^2$. 
At the initial moment of the scattering event  it is assumed that 
the height of the wall ($x_0$) and the position of $m_1$  along a period interval of the wall contour ($x_1$) are random uniformly
distributed over the respective ranges. Thus the initial position on $\mathbb{T}^2$ is a random
variable distributed according to the normalized translation-invariant measure.

Also at the initial moment of the scattering event the velocity of the wall is
assumed to be a Gaussian  random variable with zero mean and variance $\sigma_0^2$.
That is, the initial  derivative $w:=\dot{x}_0$ of $x_0$ is  normally distributed  with mean $0$ and
variance $\sigma^2=\frac{m_0}{m_1}\frac{\sigma_0^2}{a_1^2}$. Thus the probability distribution
for $w$ is given by the measure $\mu$ such that
$ d\mu(w)=(2\pi\sigma^2)^{-1/2} {e^{-\frac12w^2/\sigma^2}}\, dw.$
In the original  coordinate $v_0$ for the  velocity   of the wall,  the distribution is
$$ d\mu(v_0)= \sqrt{\frac{m_0\beta}{2\pi}}e^{-\frac{\beta}2 m_0v_0^2} \, dv_0,$$
where $\beta^{-1}:=m_0\sigma_0^2$.

 \vspace{0.1in}
\begin{figure}[htbp]
\begin{center}
%\epsfile{file=bundle.eps,scale=0.8}
\includegraphics[width=3.5in]{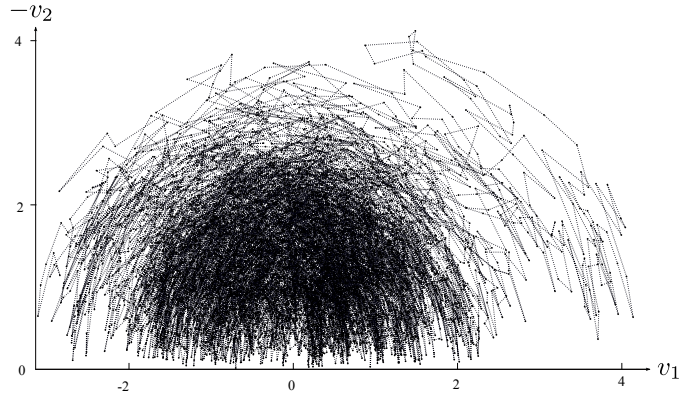}\ \ 
\caption{\small  A typical trajectory of the Markov chain for the random scattering process of
Figure \ref{movingwall}. The parameters are given in the text.
}
\label{example3traj}
\end{center}
\end{figure}

The random scattering map given in Definition \ref{steps} generates a random dynamical system on $\mathbb{H}_-^m$
whose orbits are equivalently described  as sample paths of a Markov chain process with state space  $\mathbb{H}_-^m$.
One way to interpret such multi-scattering process is to imagine that a point mass $m_1$
undergoes a random flight  inside a long  channel bounded by two parallel lines (the channel walls),
these walls having at close range (compared to the distance between the two lines) the structure depicted
in Figure \ref{movingwall}.

 Figure \ref{example3traj} shows a typical sample path of the multi-scattering Markov chain obtained numerically for
 the contour function 
 $f(z_1)=\sqrt{R^2-z_1^2}-\sqrt{R^2-a_1^2/4}$, with $a_1=1$, $R=4$,  and masses $m_0=80$ and $m_1=1$. 
 The variance is $\sigma_0=1$ and the number of iterations is $10^4$.
 (In the figure we used $-v_2$,  so  the trajectory  is shown  in the upper-half plane.)

 According to Proposition \ref{invariantproposition},
 the stationary probability distribution for $u= (\dot{x}_1, \dot{x}_2)$ is 
\begin{equation}\label{deta1}d\eta(u)=\frac1{\sigma^3\sqrt{2\pi}} u_2 e^{-\frac12 |u|^2/\sigma^2}\,  du_1\, du_2.\end{equation}
Expressed in the original velocity variables $v_1, v_2$ of $m_1$, this distribution
has the form
\begin{equation}\label{deta} d\eta(v)=\frac{\left(m_1\beta\right)^{3/2}}{\sqrt{2\pi}} v_2 e^{-\frac{\beta}2m_1 |v|^2} \, dv_1\, dv_2=\left(\sqrt{\frac{2}{\pi}}\left(m_1\beta\right)^{3/2}s^2 e^{-\frac{\beta}{2}m_1s^2}\, ds\right)\left(\frac12 \cos\theta \, d\theta\right)\end{equation}
where $s=|v|$ is  the speed of $m_1$ and $\theta$ is
the angle the velocity of $m_1$ makes with the normal to the reference plane pointing into
the region of interaction.  The fact that $\beta$ is the same in both distributions of velocities (for $m_0$ and $m_1$)
is indicative  of (thermal) equilibrium. 
These distributions are illustrated in Figure \ref{angle}.

\vspace{0.1in}
\begin{figure}[htbp]
\begin{center}
%\epsfile{file=bundle.eps,scale=0.8}
\includegraphics[width=2.7in]{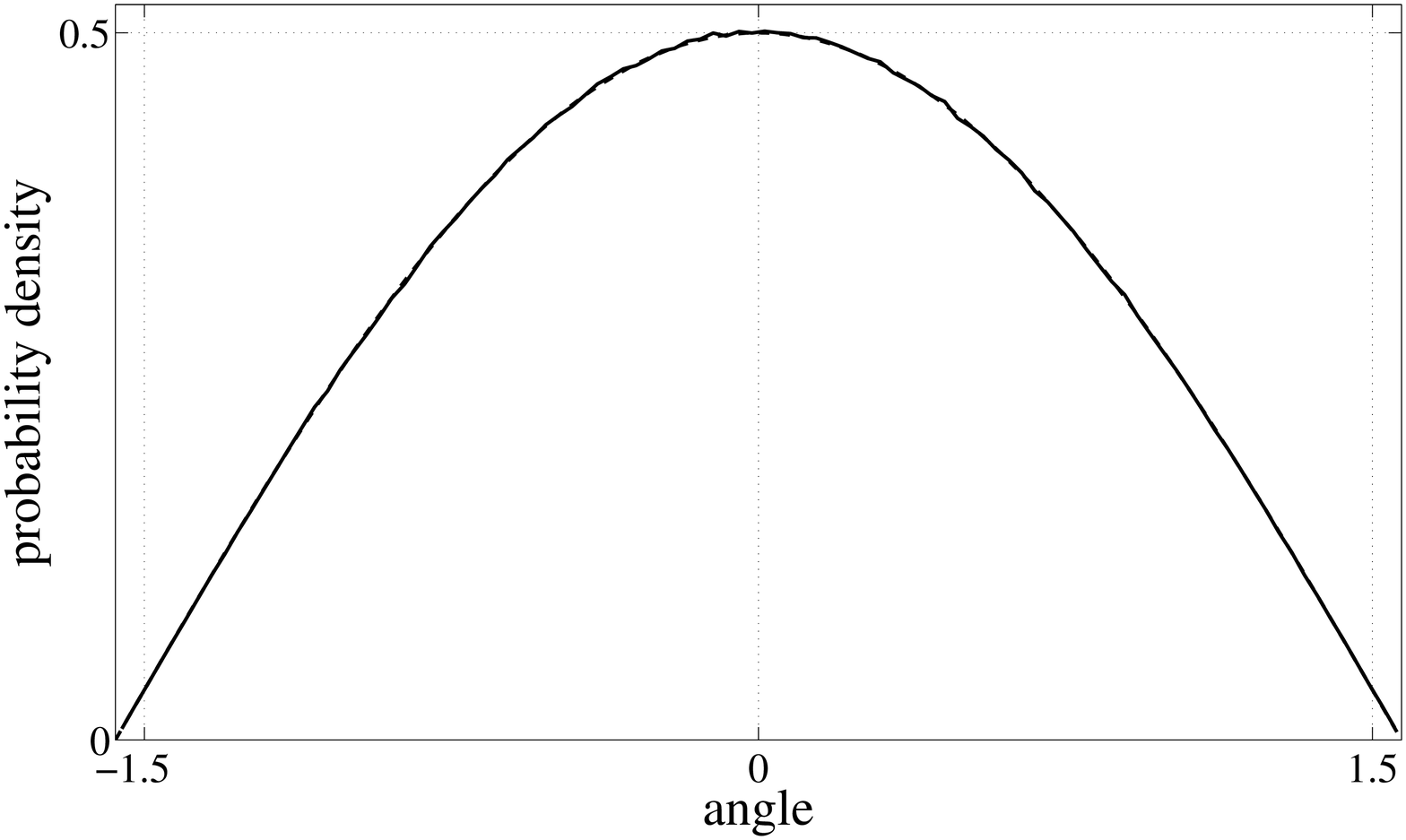}\hspace{-0.2in} \includegraphics[width=2.7in]{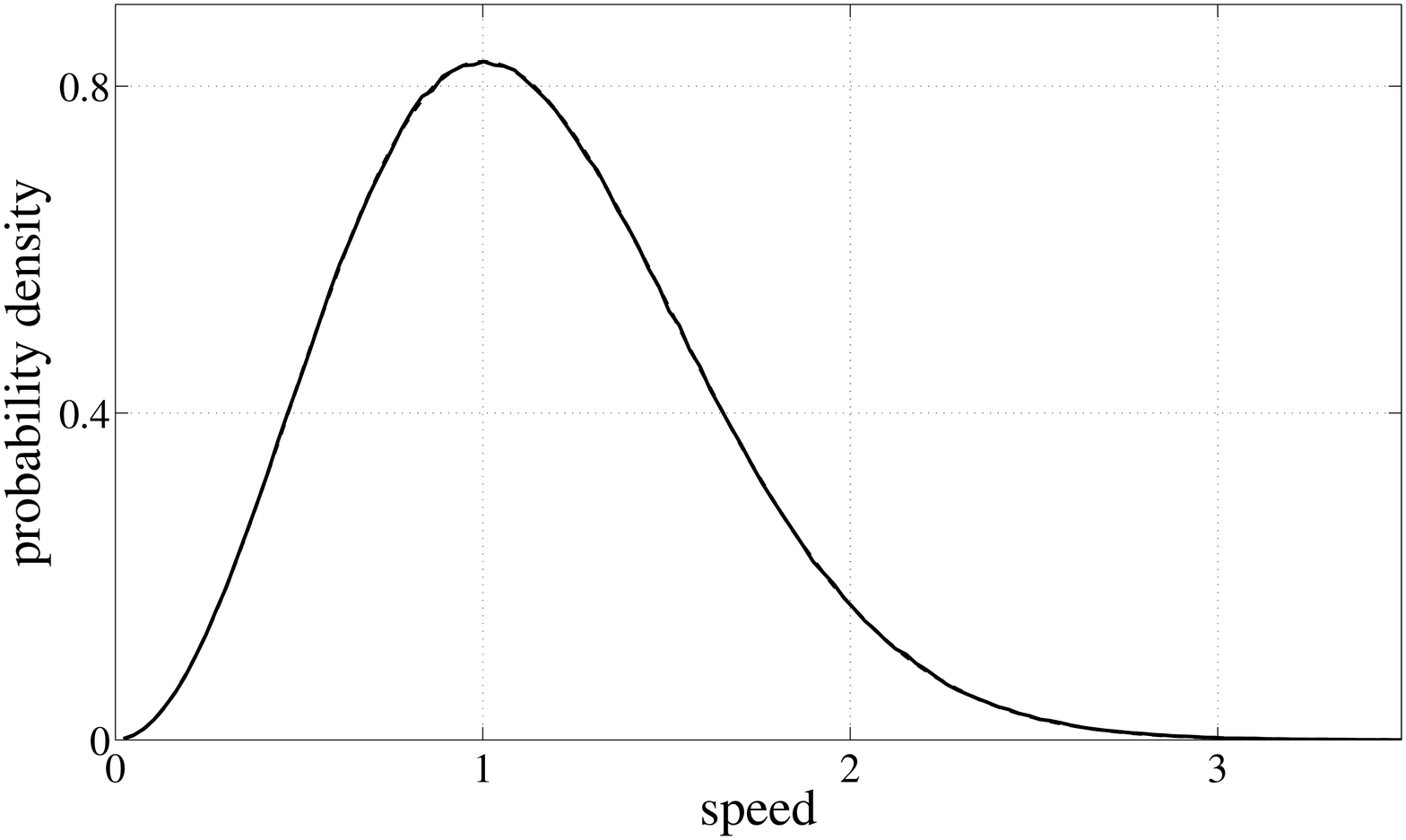}
\caption{\small  
Factor densities of the stationary distribution  $d\eta$ of scattered  angle (left) and speed (right)
  as given  on the right of  \ref{deta} above,
obtained by   numerical simulation
of the random billiard 
using a  sample run of the Markov chain of length $10^7$.
The parameters are $m_0=m_1=1$,
$\sigma_0^2=1/2$, $f(z_1)=\sqrt{R^2-z_1^2}-\sqrt{R^2-a_1^2/4}$ (an arc of circle) with $a_1=1$ and $R=3$.
 The analytically derived expressions for the distributions are also shown above (dashed lines)
 but are virtually indistinguishable  from those obtained numerically  (solid lines).}
\label{angle}
\end{center}
\end{figure}

\section{Differential approximation of the scattering operator}\label{diffapprsection}
In this section we denote  the scattering  operator  by $P_h$,
indexed by the flatness parameter $h$, and define
for all $\Phi$ in the space of compactly supported bounded functions $C_0^\infty(\mathbb{H}_-^m)$
$$\mathcal{L}_h\Phi :=\frac{P_h\Phi -\Phi}{h}.$$
Other choices of  denominator can be more natural or convenient  in specific cases, but $h$ indicates the
correct order of magnitude. 
Our immediate task is   to describe a second order differential operator $\mathcal{L}$ to which $\mathcal{L}_h$  converges  uniformly when applied
to elements of $C_0^\infty(\mathbb{H}_-^m)$.

\subsection{Definitions, notations, and preliminary remarks}
The  notations used below were introduced in Subsection 
\ref{further} and  are summarized in Figure
\ref{definitions3}. In addition, we
 occasionally use the shorthand 
 $ \xi_\smalltriangledown:=\overline{\xi}/\langle \xi, e\rangle$. 
 Also,   the variable $r=(\overline{r},c)$ will typically be used to represent the initial position 
 of trajectories, instead of $(x, c)$.  Given an initial state $(r,\xi)$ with $r$ on the reference plane,
 let $V=V(r,\xi)$ denote the component in $\mathbb{H}_-^m$ of the velocity of the return state $T(r, \xi)$.
 Recall that, at the end point,  the scattering map reflects the velocity back
into $\mathbb{H}_-^{n+1}$.
 For trajectories that collide only once with the graph of $F$,
  it will be convenient to  introduce the independent variable $x\in \mathbb{T}^{n}$
  as
 indicated in Figure \ref{definitions3},
and use it to express both $\overline{r}$ and $V$ for a given $\xi$, instead of writing 
$V(\overline{r},\xi)$ directly.  
Note that
$$\overline{r}(x,\xi)= x+(c-F(x)) \xi_\smalltriangledown,$$
whose differential in $x$ is  
$ d\overline{r}_x=I - dF_x \otimes \xi_\smalltriangledown$. By a standard determinant formula,
\begin{equation}\label{det} \det(d\overline{r}_x)=1 -  {dF_x(\xi_\smalltriangledown})=
1+
{
\langle n_\smalltriangledown(x), \xi_\smalltriangledown \rangle}. \end{equation}
Recall that $h$ is the supremum over $\mathbb{T}^{n}$ of $|\text{grad}_xF|^2$.
We wish to study the scattering process for small values of $h$.

 \begin{lemma}\label{zeta}
Let     $\xi:=v+w\in \mathbb{H}_-^m\times \mathbb{R}^{k}$ be such that the trajectory with initial state  $(\overline{r}, c, \xi)$
collides with the graph of $F$ only once for all  $\overline{r}\in \mathbb{T}^{n}$.
We regard $\overline{r}$ as a function of the initial velocity and a point $x\in \mathbb{T}^n$, as indicated in Figure
\ref{definitions3}.
Let $V=V(x,v,w)$ be the component in $\mathbb{H}_-^m$ of the velocity of the billiard trajectory as it returns to
the reference plane after one iteration of the scattering map. 
Then
  $V=v+2\zeta_1 - 2\zeta_2$, where
\begin{align*}
\zeta_1(x,v,w)&:=\langle n(x), e\rangle \left(\langle \overline{n}(x), {v}\rangle e + 
 \langle \overline{n}(x), w\rangle  e -\langle v,e\rangle \overline{n}^\curlyvee(x)\right)\\
\zeta_2(x,v,w)&:= \langle \overline{n}(x), {v}\rangle \overline{n}^\curlyvee(x) + \langle\overline{n}(x),w\rangle \overline{n}^\curlyvee(x) + |\overline{n}(x)|^2 \langle v,e\rangle e. 
\end{align*}
If $F$ is symmetric, these functions satisfy
 $\zeta_1(-x,v,w)=-\zeta_1(x,v,w)$,   $\zeta_2(-x,v,w)=\zeta_2(x,v,w)$.
\end{lemma}
\begin{proof}
This is an  entirely  straightforward calculation, of  which we indicate a few steps. The reflection of $\xi$ after the single collision with the graph of $F$ at $x\in \mathbb{T}^{n}$ is naturally given by  $\xi'=\xi -2\langle n(x),\xi\rangle n(x) \in \mathbb{H}_+^{n+1}$. 
This is then reflected by a plane perpendicular to $e$, resulting in 
$$\eta=\xi'- 2\langle \xi',e\rangle e = \xi +2\langle n,e\rangle\left(\langle \overline{n}, \overline{\xi}\rangle e - \langle \xi, e\rangle \overline{n} \right)  -2\left(\langle \overline{n}, \overline{\xi} \rangle \overline{n} + |\overline{n}|^2 \langle\xi,e\rangle e\right)\in \mathbb{H}_-^{n+1}.$$
Now apply the linear  projection ${\eta}\mapsto \eta^\curlyvee$  and use that $\langle \xi, e\rangle=\langle v, e\rangle$
and $\langle\overline{n},\overline{\xi}\rangle =\langle\overline{n}, v\rangle + \langle \overline{n},w\rangle$ to obtain the stated   identity relating $V$ and $v$.  For the rest, use $\langle n(-x), e\rangle = \langle n(x), e\rangle $ and 
$\overline{n}(-x)=-\overline{n}(x)$.
 \end{proof}

 \begin{lemma}\label{lemmaD0} Define $W(v,h):=-|v|+{|\langle v,e\rangle|}/{4\sqrt{h}} $.
Then for small enough $h$ (e.g., $h\leq (3/4)^2$),
 for all $x\in \mathbb{T}^{n}$  and all $\xi$ in the set
 $$\mathcal{D}_h:=\left\{v+w\in \mathbb{H}^{n+1}_-: |w|<W(v,h)\right\}$$  the trajectory  with initial vector $\xi$ 
   starting at   $(\overline{r}(x,\xi),c)$ 
   collides with the graph of $F$ only once and the Jacobian determinant  of $x\mapsto \overline{r}(x,\xi)$ satisfies $\det(d\overline{r}_x)=1+\langle n_\smalltriangledown(x), \xi_\smalltriangledown\rangle>0$. 
 \end{lemma}
 \begin{proof} Let  $\xi':=\xi-2\langle\xi,n(x)\rangle n(x)$.
A  sufficient condition for   single collision  is  
$ {|\langle \xi',e\rangle|}/{|\overline{\xi}'|}>\sqrt{h}.$
In fact, if there is a second collision elsewhere on the graph of $F$ under this condition,
a comparison of slopes would indicate the existence of a point where the gradient of $F$ exceeds $\sqrt{h}$, a contradiction. 
 Using   $|\xi'|=|\xi|$ and 
simple algebra, this is equivalent to $$\langle\xi',e\rangle^2>(h/(1+h)) |\xi|^2.$$
Further elementary manipulations give
$ \langle \xi',e\rangle=-\langle v,e\rangle    + 2\langle v, e\rangle |\overline{n}|^2 - 2\langle {v +w}, \overline{n}\rangle \langle n, e\rangle.$
From  $|\langle n, e\rangle|\leq 1$,  $|\langle {v+w}, \overline{n}\rangle|\leq( |v|+|w|)|\overline{n}|$, 
    and $|\overline{n}|\leq \sqrt{h}<1$,    we derive
$$|\langle \xi',e\rangle| \geq   |\langle v,e\rangle | - 2\sqrt{h}(|v|+ |w|)    -2h|\langle v,e\rangle |.$$
It follows that  
$$(1-2h) |\langle v,e\rangle |  - 2\sqrt{h}(|v|+|w|)> \sqrt{h}\sqrt{|v|^2 + w^2}$$
is also a   sufficient condition for single collision.
Since $0<\sqrt{x^2+y^2} \leq |x|+|y|$, yet another sufficient condition  is
\begin{equation}\label{ineq}
|w|<(1-2h) \frac{|\langle v,e\rangle|}{3\sqrt{h}}-|v|.\end{equation}
The inequality $1+\langle n_{\smalltriangledown}, \xi_\smalltriangledown\rangle >0$ can be rewritten as
$ \langle n,e\rangle |\langle v,e\rangle | > |\langle \overline{n}, {v+w}\rangle |,$
which is easily seen to be implied by 
$$|w|<\frac{\sqrt{1-h}|\langle v, e\rangle |}{\sqrt{h}}-|v|. $$
But this in turn is implied by inequality \ref{ineq} for sufficiently small $h$. For small enough $h$, we may simplify
\ref{ineq} by writing the right-hand side as $|\langle v,e\rangle|/4\sqrt{h} - |v|$.
 \end{proof}

 \subsection{The operator approximation argument}\label{appargument}
Let $\Phi$ be a smooth function defined on a subset $U\subset \mathbb{R}^m$.
 The {\em $k$th differential} $d^k\Phi_v$   of   $\Phi$ at  $v\in U$ is 
 the symmetric $k$-linear map on  $T_vU$  such that
 $ d^k\Phi_v(e_{i_1}, \dots, e_{i_k})= (D_{i_1}\cdots D_{i_k}\Phi)(v),$
  where $D_i$ is the directional  derivative in the direction of the constant  vector field $e_i$. 
  If $\xi$ is a constant vector field, then $$d^k\Phi_v(\xi,\dots, \xi)=\left.\left(\frac{d}{ds}\right)^k\right|_{s=0}\Phi(v+s\xi).$$
 Let  $g(s)=\Phi(v+s\xi)$. In  the above  notations,  the Taylor approximation of $g(1)$
  up to degree $2$, expanded in derivatives of $g(s)$  at $s=0$, has  the form
  \begin{equation}\label{Taylor} \Phi(v+\xi)=\Phi(v) + d\Phi_v(\xi)+\frac12 d^2\Phi_v(\xi,\xi)+R_v(\xi)\end{equation}
  where 
  $|R_v(\xi)|=\left|\int_0^1\frac{(1-t)^2}{2}(d^3\Phi)_{v+t\xi}(\xi,\xi,\xi)\, dt\right|\leq \frac16 \|d^3\Phi\| |\xi|^3.$
  For the main theorem below, where $\Phi$ will be compactly supported in $\mathbb{H}_-^m$,
   $\|d^3\Phi\|$ may be taken  to be the supremum over $v$ of any choice of  norm on the $3$-linear map at $v$.

We introduce   linear maps   $C:\mathbb{R}^{n+1}\rightarrow \mathbb{R}^k$
and $A:\mathbb{R}^{n+1}\rightarrow \mathbb{R}^{n}$ defined by
$$C:=\int_{\mathbb{R}^k} w^*\otimes w\, d\mu(w),\ \ A:= \int_{\mathbb{T}^n} \overline{n}^*(x)\otimes \overline{n}(x)\, d\lambda(x).$$
Thus,
 by definition,
$Cu=\int_{\mathbb{R}^k} \langle w, u\rangle w\, d\mu(w)$,
 and 
$A:=\int_{\mathbb{T}^n}A(x)\, d\lambda(x)$, where $A(x)$ is the linear map $A(x)u:=\langle u, \overline{n}(x) \rangle\, \overline{n}(x)$.
Then $A$  and $C$ are  non-negative definite symmetric linear maps.

For convenience of notation, we  shall often write  below $E[\cdots]=\int_{\mathbb{T}^n}\cdots\, d\lambda(x)$.
Given a twice differentiable function $\Phi$, let $\text{Hess}_v \Phi$ represent the matrix associated via  the standard inner product $\langle \cdot, \cdot \rangle$ to the second derivatives quadratic form $d^2\Phi_v$.
Let $Q:\mathbb{R}^{n+1}\rightarrow \mathbb{R}^q$ be any self-adjoint  map and denote $A^Q=QAQ$. 
Observe  the identities:
$$
\text{Tr}\left(A^Q\right)=E\left[\left| Q\overline{n}\right|^2\right], \  \text{Tr}\left(A^Q\text{Hess}_v \Phi\right)=E\left[d^2\Phi_v(Q\overline{n},Q\overline{n} )\right], \   \left\langle A^Qu,u\right\rangle =  E\left[\langle Q\overline{n} , u\rangle^2 \right]   $$
as well as 
$$
d\Phi_v A^Q u=E\left[\langle Q\overline{n},u\rangle\, d\Phi_v Q\overline{n} \right],\ 
 d^2\Phi_v\left(A^Qv_1,v_2\right)=E\left[\langle Q\overline{n}, v_1\rangle\, d^2\Phi_v(Q\overline{n},v_2)\right]. 
$$
 Similar identities hold for $C$. In particular,
 $$\text{Tr}(AC)=\text{Tr}\left(A^{C^{1/2}}\right)=\int_{\mathbb{R}^k}E\left[\langle \overline{n},w\rangle^2\right]\, d\mu(w).$$
 The orthogonal 
  projection $\eta\mapsto \eta^{\curlyvee}$
 will be indicated by 
 $Q^\curlyvee:\mathbb{R}^{n+1}\rightarrow \mathbb{R}^{m}$.  If  $\Phi$ is a function
 on $\mathbb{H}_-^m$ which  does not depend on $w\in \mathbb{R}^k$, 
then 
$\text{grad}_v\Phi=Q^\curlyvee\text{grad}_v\Phi$
and $\text{Hess}_v\Phi=Q^\curlyvee (\text{Hess}_v\Phi) Q^\curlyvee$.
Similarly, we may define the orthogonal projection $Q^\curlywedge$ to $\mathbb{R}^k$.
 if  $\mu$ is the Gaussian distribution with temperature parameter $\sigma^2$ (see
Subsection \ref{invariantstandard}) then
  $ C:=\int_{\mathbb{R}^k} w^*\otimes w \, d\mu(w) =\sigma^2 Q^\curlywedge.$
Observe that $A$ goes to $0$ linearly in $h$. The following assumption
 is very commonly  satisfied:
 \begin{assumption}\label{assumption1}
We  suppose that the limit 
$ \Lambda:=\lim_{h\rightarrow 0}A/h$
exists.
\end{assumption}

\begin{theorem}\label{maintheorem}
Let  $\mu$ be a probability measure on $\mathbb{R}^k$ with mean $0$,  finite second moments given by
the matrix $C$ and finite third moments.  Under   Assumption \ref{assumption1}, define the differential
operator
\begin{align*}(\mathcal{L}\Phi)(v)& =  
-4 \left\langle \Lambda\, \text{\em grad}_v \Phi, v\right\rangle  
 +  \frac{2}{\langle v, e\rangle} \left[ \langle \Lambda  v, v\rangle +\text{\em Tr}\left(C\Lambda\right) -  \text{\em Tr}( \Lambda)\langle v,e\rangle^2\right] \langle \text{\em grad}_v\Phi,e\rangle +  \\   
  &   2\langle v,e\rangle \left[ \langle v,e\rangle \text{\em Tr}\left(\Lambda\,  \text{\em Hess}_v \Phi\right) -2 \langle  \Lambda\,  \text{\em Hess}_v\Phi\, e,    v\rangle \right]  +2\left(\langle \Lambda  v, v\rangle + \text{\em Tr}\left(C\Lambda\right)\right)\langle \text{\em Hess}_v  \Phi \, e,e\rangle. 
\end{align*}
on smooth functions $\Phi$. Then
$$\lim_{h\rightarrow 0}\frac{P_h\Phi - \Phi}{h}= \mathcal{L}\Phi $$
uniformly on $\mathbb{H}_-^m$, for each $\Phi\in C_0^\infty(\mathbb{H}_-^m)$.
When $\mu$ is the Gaussian distribution on $\mathbb{R}^k$ with temperature parameter $\sigma^2$ (see Subsection
\ref{invariantstandard}), then  $C=\sigma^2 Q^\curlywedge$. 
\end{theorem}
 \begin{proof}
Recall that  the translation invariant probability measure on $\mathbb{T}^n$ is here denoted by $\lambda$.
 With  the notations of Lemma \ref{lemmaD0} in mind we define
 $\mathcal{D}_h(v):=\{w\in \mathbb{R}^k: |w|<W(v,h)\}.$
Then for  $\Phi\in C_0^\infty(\mathbb{H}_-^{m})$,  
 $$(P_h\Phi)(v)=\int_{\mathbb{R}^k}\int_{\mathbb{T}^{n}} \Phi(V(\overline{r}, v, w))\,d\lambda(\overline{r})\, d\mu(w)=I_1+I_2,$$ 
 where  for  $I_1$ the   integration in $w$ is over $\mathcal{D}_h(v)$, and for $I_2$ the integration is over  $\mathcal{D}_h^c(v)$.
 Notice that
 $|I_2|\leq (1-\mu(\mathcal{D}_h(v)))\|\Phi\|_\infty$ goes to  $0$  as $h$ approaches $0$. 
 
 We now concentrate on $I_1$.
 Using Lemma \ref{zeta} and the form of the Jacobian determinant  $\det(d\overline{r}_x)$ given in Lemma \ref{lemmaD0},
 $$I_1=\int_{\mathcal{D}_h(v)} \int_{\mathbb{T}^{n}}\Phi(v+2\zeta_1-2\zeta_2)\left(1+\delta(x,v,w)\right)\,  d\lambda(x)\, d\mu(w),$$
where
$\delta(x,v,w):={\langle \overline{n}(x),{v}+w\rangle}/ \langle n(x),e\rangle\langle v,e\rangle$
and $\zeta_i=\zeta_i(x,v,w)$. To simplify the notation we write
$I_1=\int_{\mathcal{D}_h(v)}I_1(v,w)\, d\mu(w)$, where $$I_1(v,w):=E[\Phi(v+2\zeta_1-2\zeta_2)(1+\delta)],$$ and $E$, defined earlier, indicates
average over $x$. We now use the symmetries: $$\zeta_1(-x,v,w)=-\zeta_1(x,v,w), \ \zeta_2(-x,v,w)=\zeta_2(x,v,w), \ \delta(-x,v,w)=-\delta(x,v,w)$$
to write
$$I_1(v,w)=E\left[\frac{\Phi(v+2\zeta_1-2\zeta_2)+\Phi(v-2\zeta_1-2\zeta_2)}{2}+\frac{\Phi(v+2\zeta_1-2\zeta_2)-\Phi(v-2\zeta_1-2\zeta_2)}{2}\delta\right].$$
 Notice that $\zeta_i$ are of the order $O(h)$ in $h$ for each $v$ and $w$. 
Each $\Phi(v+\eta)$ may be approximated by a Taylor polynomial at $v$ up to degree $2$ (\ref{Taylor}),
$$ \Phi(v+\eta)=\Phi(v)+d\Phi_v \eta + \frac12 d^2\Phi_v(\eta,\eta)+R_v(\eta),$$
where $|R_v(\eta)|\leq \frac16 \|d^3\Phi\| \|\eta\|^3.$
The sum of all terms inside $E[\cdots]$
has second degree Taylor polynomial of the form
$$P_2(v,\zeta_1,\zeta_2)= \Phi(v) +2\left( d\Phi_v(-\zeta_2 +\zeta_1\delta)  + d^2\Phi_v(\zeta_1,\zeta_1) + d^2\Phi_v(\zeta_2,\zeta_2) -2 d^2\Phi_v(\zeta_1,\zeta_2)\delta\right).$$
Keeping only terms   in  $I_1(v,w)$ up to first degree in $h$  yields
$$I_1(v,w)=\Phi(v) +2E\left[ d\Phi_v(-\zeta_2 +\zeta_1\delta)\right]  +2E\left[ d^2\Phi_v(\zeta_1,\zeta_1) \right]  +\text{Error}(v,w,h),$$
where the error term is bounded by a product, $|\text{Error}|\leq C_{\Phi} p_3(|v|,|w|) h^{3/2}$;  here  $C_\Phi$ is a constant depending  only on
the derivatives of $\Phi$ up to third order and  $p_3$ is a polynomial in $|v|, |w|$ of degree at most $3$ that does not depend on $\Phi$
and $h$.
The linear term in $\zeta_i$ contributes to $I_1(v,w)$ the expression
\begin{align*} -4E\left[\langle \overline{n},v\rangle d\Phi_v \overline{n}^\curlyvee \right]  +&
\frac{2}{\langle v,e\rangle} E\left[\langle \overline{n},v\rangle^2 + \langle \overline{n},w\rangle^2  -|\overline{n}|^2\langle v,e\rangle^2\right] d\Phi_v e  +\\
&\frac{4}{\langle v,e\rangle}E\left[  \langle \overline{n}, w\rangle \langle v,e\rangle d\Phi_v\overline{n}^\curlyvee    
 - \langle \overline{n}, w\rangle \langle \overline{n},v\rangle   d\Phi_v e
\right].
\end{align*}
Since the measure $\mu$ is assumed to have mean $0$ (and finite second and third  moments),
the last term above (in which   $w$ appears linearly) vanishes after integration over $\mathcal{D}_h(v)$.
Therefore, the zeroth and first order terms (in $\Phi$) contribution to $I_1$
are
$$I_1=\alpha_h\left\{\Phi(v)
-4E\left[\langle \overline{n},v\rangle d\Phi_v \overline{n}^\curlyvee \right]  +\frac{2}{\langle v,e\rangle} E\left[\langle \overline{n},v\rangle^2 +{\alpha_h^{-1}}\langle C_h\overline{n},\overline{n}\rangle -|\overline{n}|^2\langle v,e\rangle^2\right] d\Phi_v e 
\right\}  +\cdots $$
where  $\alpha_h:=\mu(\mathcal{D}_h(v))$ goes to $1$ and $C_h:=\int_{\mathcal{D}_h(v)} w^*\otimes w \, d\mu(w)$ goes to $C$ as $h$ approaches $0$.

We now proceed to the second order terms. A similar kind of analysis, where we disregard first order terms in $w$
and drop  terms in $h$ of power $3/2$ or greater  into the error term (this involves 
approximating  an overall multiplicative factor $\langle n,e\rangle^2$ by $1$),  yields the second order (in $\Phi$) contribution to 
$I_1$ given by the sum  $a_1+a_2$, where (separately, so as to fit in one line)
\begin{align*}
a_1&=2 \alpha_h\left\{\langle v,e\rangle^2 E\left[ d^2\Phi_v(\overline{n}^\curlyvee, \overline{n}^\curlyvee)\right]-2\langle v,e\rangle E\left[ \langle \overline{n}, v \rangle d^2\Phi_v(\overline{n}^\curlyvee,e)\right]\right\}\\
a_2&=2\alpha_h  E\left[\langle \overline{n},v\rangle^2 \right] d^2\Phi_v(e,e)+2 E\left[ \langle C_h\overline{n},\overline{n}\rangle\right] d^2\Phi_v(e,e).
\end{align*}
 Collecting all terms, and using the identities listed for  $A$ and $C$ noted prior to the statement of the theorem, yields
\begin{align*}
\alpha_h^{-1}I_1 =\Phi(v) 
-4\left \langle   \text{grad}_v\Phi, A v\right\rangle  &
 +  \frac{2}{\langle v, e\rangle} \left[ \langle Av, v\rangle +\alpha_h^{-1}\text{Tr}\left(C_hA\right) -\text{Tr} A\langle v,e\rangle^2\right] \langle \text{grad}_v\Phi,e\rangle     \\
 & + 2\langle v,e\rangle \left[ \langle v,e\rangle \text{Tr}\left(A\circ \text{Hess}_v \Phi\right) -2\left\langle \text{Hess}_v \Phi  A  v,e\right\rangle \right] \\
 & +2\left(\langle Av, v\rangle + \alpha_h^{-1} \text{Tr}\left({C_h}A\right)\right)\langle \text{Hess}_v \Phi\,  e,e\rangle + \text{Error}(v,h)
\end{align*}
where the error term is of   order   $h^{3/2}$. We have used that the third moment of $\mu$ is finite to ensure that
the error term is finite.
If $\Phi\in C_0^\infty(\mathbb{H}_-^{m})$, it follows that as $h\rightarrow 0$,  the quantity   $(I_1-\Phi(v))/h$ has the same
limit as $((P_h\Phi)(v)-\Phi(v))/h$, which is $(\mathcal{L}\Phi)(v)$, the convergence is uniform, and
the limit is $(\mathcal{L}\Phi)(v)$ as claimed.
\end{proof}

Recall that $\mathbb{H}_-^{n+1}=\mathbb{H}_-^m\times \mathbb{R}^k$
is the decomposition of velocity space into ``observable'' and ``hidden'' components, with respective projections  $Q^\curlyvee$ and $Q^{\curlywedge}$ defined earlier.
Let  $A^\curlyvee = Q^\curlyvee A Q^\curlyvee$ and $A^\curlywedge = Q^{\curlywedge}A Q^\curlywedge.$
We make
now an  additional but  very natural  assumption, which   holds in  all the examples discussed in this paper,
that $\Lambda$ is {\em adapted}, according to the  following definition.
\begin{definition}\label{adapted}
The linear map $A$  is  {\em adapted}  
if
 $A=A^\curlyvee+A^\curlywedge$, in which case
a similar decomposition holds for $\Lambda$ under Assumption \ref{assumption1},
and we say that   $\Lambda$ is also  {adapted}.
\end{definition}

For adapted $\Lambda$ and for $C$ and $\sigma^2$ as described at the end of Theorem \ref{maintheorem}, 
$\sigma^2=\text{Tr}(C\Lambda)/\text{Tr}(\Lambda^\curlywedge).$
Also recall the stationary measure $d\eta(v)=\rho(v)\, dV(v)$ described in Proposition \ref{invariantproposition},
whose density  is  
$\varrho(v)=c v_m \exp\left(-\frac12|v|^2/\sigma^2\right)$, where   $c$
is a   constant of normalization.

\begin{corollary}\label{kbig}
Let the same assumptions of Theorem \ref{maintheorem} hold. Further suppose that $k\geq 1$ and that  $\Lambda$ is adapted.
Let 
 $e_1, \dots, e_{m-1}, e=e_m\in \mathbb{R}^m$ be an orthonormal basis of eigenvectors   of
$\Lambda^\curlyvee$, with $\Lambda^\curlyvee e_i= \Lambda e_i = \lambda_i e_i$, and $\lambda_m=0$. 
   The partial derivative of a function $\Phi$ on $\mathbb{H}_-^m$ in the direction $e_i$ is denoted $\Phi_i$
   and the coordinate functions are $v_i:=\langle v, e_i\rangle.$ Then, for $\Phi\in C_0^\infty(\mathbb{H}_-^m)$,
  $$ \left(\frac12\mathcal{L}\Phi\right)(v)=
   \left(
   \sum_{i=1}^{m-1}\lambda_i v_i^2 +\text{\em Tr}(C\Lambda)\right)
   \left[\left(\frac1{v_m} - 
\frac{v_m}{\sigma^2}\right)\Phi_m(v) 
 +\Phi_{mm}(v)\right]
 +\sum_{i=1}^{m-1}\lambda_i\left(\mathcal{L}_i\Phi\right)(v)$$
 where $\mathcal{L}_i$ is defined by
 $$ (\mathcal{L}_i\Phi)(v)=-2v_i\Phi_i(v)
 +  v_m^2\Phi_{ii}(v)-2 v_iv_m \Phi_{im}(v)  -\left[1-
 \left(\sigma^2\text{\em Tr}( \Lambda^\curlyvee)\right)^{-1}\sum_{j=1}^{m-1}\lambda_jv_j^2\right]v_m \Phi_m(v).$$
 This rather cumbersome expression can be greatly simplified by the following coordinate change:
 $x_i:=v_i$ for $i=1,\dots, m-1$ and $x_m:=|v|^2/2\sigma^2$.  Let $h(x)=2\sigma^2 x_m-x_1^2-\dots-x_{m-1}^2.$
 Then
 $$ \left(\frac12\mathcal{L}\Phi\right)(x)=\sum_{i=1}^{m-1}\lambda_i \left(h(x)\Phi_i\right)_i + \frac{\text{\em Tr}(C\Lambda)}{\sigma^4}e^{x_m}\left(h(x)e^{-x_m}\Phi_m\right)_m$$
 where $\Phi$ is a compactly supported function on $\{x:2\sigma^2x_m> x_1^2+\dots+x_{m-1}^2\}$.
\end{corollary}
\begin{proof}
This is  derived from Theorem \ref{maintheorem} by  straightforward  calculations. 
\end{proof}

As a special case,  suppose that  $n=0$. Then $m=1$ and 
 $\mathbb{H}_-^m=(-\infty,0)$, in the direction of the single vector $e$. Write  $\Lambda=\lambda>0$ and $C=\sigma^2>0$.  Here, only the speed, $v\in (0,\infty)$, is of interest.
 We denote by $\Phi'$ and $\Phi''$ the first and second derivatives with respect to $v$. 
  Then 
  \begin{corollary}[Dimension $1$]\label{cordim1} Under the assumptions of Theorem \ref{maintheorem} and that $n=m=k=1$, then for
  any compactly supported smooth function $\Phi$ on $(0,\infty)$,  
  \begin{equation}\label{operator}(\mathcal{L}\Phi)(v)=2\lambda\sigma^2\left[\left(\frac{1}{v}-\frac{v}{\sigma^2}\right)\Phi' + \Phi''\right]. \end{equation}
  This can   be written in Sturm-Liouville form as
  $$\frac{1}{2\lambda\sigma^2}(\mathcal{L}\Phi)(v)=\frac1\varrho \frac{d}{dv}\left(\varrho \frac{d\Phi}{dv}\right)$$
where $$ \varrho=\sigma^{-2}v\exp\left(-\frac{v^2}{2\sigma^2}\right).$$
\end{corollary}
\begin{proof} This is  a straightforward consequence of Corollary \ref{kbig}. Note that the coordinates $v_i$
are absent  for $i=1,\dots, m-1$
and $\text{Tr}(C\Lambda^\curlywedge)=\lambda\sigma^2$.
\end{proof}

Consider now the case  $k=0$, or $m=n+1$. This means that only the initial position in
$\mathbb{T}^n$ is random,  while the initial velocity is fully specified.
Then, as the speed $|v|$ of the billiard trajectory does not change 
after collision, we may restrict the state space of the Markov operator $P$ to the hemisphere of radius $\rho:=|v|$ in
$\mathbb{H}_-^{n+1}$. This hemisphere is diffeomorphic to the ball $D^{n}_\rho$ of radius $\rho$, via the
linear   projection $Q:\mathbb{R}^{n+1}\rightarrow \mathbb{R}^n$ taking $e$ to $0$ and fixing the other coordinate vectors.
In this special case, we can restrict attention to functions of the form $\Phi=\Psi\circ Q$, where
$\Phi(\overline{v})$ is a smooth function on $D^n_\rho$ and $\overline{v}=Qv$. 
For these functions, $\langle \text{grad}_v\Phi , e\rangle=0$ and 
$ \langle \text{Hess}_v \Phi\,   u_1,u_2\rangle=0$ if either $u_1$ or $u_2$ or both are  multiples of  $e$. 
Thus the operator $\mathcal{L}$  reduces to
 $$( \mathcal{L}\Psi)(\overline{v})=-4\langle Q\,  \text{grad}_{\overline{v}} \Psi, \Lambda  \overline{v}\rangle + 2\left(\rho^2 -|\overline{v}|^2\right) \text{Tr}(\Lambda\,  \text{Hess}_{\overline{v}}^Q \Psi).$$

\begin{corollary}[Constant speed] \label{constspeed}
Let the same assumptions of Theorem \ref{maintheorem} hold, and that $k=0$.
Without loss of generality,  let the particle speed be $1$.
Let $\lambda_i\geq 0$, $i=1, \dots, n$ and $e_i$ be
as in Corollary \ref{kbig}, while  $v_i$ is now used as the  coordinates on $D_1^n$ whose coordinate vector fields are the $e_i$. 
In this new   system the operator $\mathcal{L}$ has the Sturm-Liouville form
\begin{equation}\label{genLegendreOp} (\mathcal{L}\Psi)(v) = 2\sum_{i=1}^n \lambda_i  \left((1-|v|^2) \Psi_i\right)_i\end{equation}
where the index in $\Psi_i$ indicates partial derivative in $v_i$.
In dimension $n=1$,   $\mathcal{L}$ is    the standard Legendre's differential operator
on the interval $[-1,1]$ up to a multiplicative constant. 
 \end{corollary}
\begin{proof}
This readily follows from the general form of the operator.
\end{proof}

When $k\geq 1$, define the inner product   
 $$\langle \Phi, \Psi\rangle := \int_{\mathbb{H}_-^m} \Phi(v)\Psi(v) \varrho(v) \, dV(v)$$
on
 compactly supported smooth functions. When $k=0$, we restrict the functions to
 the unit hemisphere equipped with the measure (given by a scalar multiple of) $\langle v,e\rangle \, d\omega(v)$,
 where $\omega$ is the Euclidean  volume measure on the hemisphere, and define the inner product accordingly.
 (In this latter case, the density of the measure is proportional to the cosine of the angle between the vector $v$ and
 the unit normal to the boundary of $\mathbb{H}_-^m$.)
 We say that $\mathcal{L}$ is {\em symmetric} if 
$\langle \mathcal{L}\Phi, \Psi\rangle = \langle \Phi, \mathcal{L} \Psi \rangle$.

\begin{theorem}\label{elliptic}
Under the general conditions of Theorem \ref{maintheorem}, assume further that $\Lambda$ is adapted and
positive definite. (Recall that $\Lambda$  is in general non-negative definite.) Then $\mathcal{L}$ is a second order, symmetric, elliptic
operator on $C_0^\infty(\mathbb{H}_-^m)$.
\end{theorem}
 \begin{proof}
The claims are obtained by a long but completely straightforward calculation.
 We only check ellipticity for $k\geq 1$. (The case $k=0$ is even simpler.)
Recall that the symbol of the second order operator $\mathcal{L}$ is the quadratic form
$\sigma_{\mathcal{L}}(\xi)=\sum a_{ij}(v)\xi_i\xi_j$, where the $a_{ij}(v)$ are the coefficients of the
second order terms of $\mathcal{L}$ and $\xi$ is a vector of dimension $m$. 
Starting from  the expression of $\mathcal{L}$ given  in   Corollary \ref{kbig},
the symbol can the written in the  form
 $$\sigma_{\mathcal{L}}(\xi)= 2\sum_{i=1}^{m-1}\lambda_i \left(v_m\xi_i-v_i\xi_m\right)^2 + 2 \sigma^2 \text{Tr}(\Lambda^\curlywedge) \xi_m^2.$$
 Since $\lambda_i>0$ for $i=1, \dots, m$, and both   $\sigma^2>0$ and $v_m>0$, then  $\sigma_{\mathcal{L}}(\xi)=0$
 only if $\xi=0$. 
 \end{proof}
 
 That $\mathcal{L}$ is symmetric and elliptic can be seen more easily by noting that
 it can be put in Sturm-Liouville form relative to the MB-distribution $\varrho$. To see this, we
 first introduce the following first order differential operators  in $\mathbb{R}^m$. (The subindex $m$ in
 $\Phi_m$ and $\langle\cdot,\cdot\rangle_m$ indicates derivative
 in the direction $e=e_m$.) For a smooth function $\Phi$,
 $$(\mathcal{D}\Phi)(v):= \sqrt{2}\left[\Lambda^{1/2}\left(v_m\, \text{grad}_v\, \Phi- \Phi_m(v) v\right) +\text{Tr}\left(C\Lambda\right)^{1/2}\Phi_m(v) e\right].$$
If $\Xi$ is a vector field in $\mathbb{R}^m$,  
$$ (\mathcal{D}'\Xi)(v):=\sqrt{2}\left[- \text{div}\left(v_m\, \Lambda^{1/2}\Xi\right)+\langle v,  \Lambda^{1/2}\Xi\rangle_m - \text{Tr}\left(C\Lambda\right)^{1/2}\langle \Xi, e\rangle_m\right].$$
Then $\mathcal{D}'$ is the  adjoint of $\mathcal{D}$ relative to the Lebesgue measure on $\mathbb{R}^m$. That is, if either $\Phi$ or $\Xi$ is compactly supported, then
$$ \int_{\mathbb{R}^m}\left(\mathcal{D}\Phi\right)\, \Xi\, dV= \int_{\mathbb{R}^m}\Phi \left(\mathcal{D}'\Xi\right)\, dV.$$
We now restrict these operators to the half-space $\mathbb{H}_-^m$ and
define 
$ \mathcal{D}^*\Xi := \varrho^{-1}\mathcal{D}'\left(\varrho\, \Xi\right).$
Clearly, $\mathcal{D}^*$ is the adjoint of $\mathcal{D}$ with respect to the density $\varrho$:
$$\int_{\mathbb{H}_-^m} \left(\mathcal{D}^*\Xi\right) \Phi\, \varrho \, dV=  \int_{\mathbb{H}_-^m}\Xi\cdot \left(\mathcal{D} \Phi\right)\, \varrho \, dV $$
\begin{proposition}
Under the assumptions of Theorem \ref{maintheorem} and that $\Lambda$ is adapted, the differential
operator $\mathcal{L}$ has the  form
$$\mathcal{L}\Phi=-  \mathcal{D}^* \mathcal{D}\Phi$$
where $\Phi$ is a smooth, compactly supported function in $\mathbb{H}_-^m$.
\end{proposition}
\begin{proof}
This amounts to a tedious but entirely straightforward exercise.
\end{proof}
\section{Diffusion limits of the iterated scattering chains }
One reason for relating the Markov operator $P$ to an elliptic second order differential operator is
the desire to obtain diffusion approximations of  Markov chains associated to our random mechanical models.
In this section we turn to such   approximations. 

\subsection{Generalities about diffusion limits}
The results stated here are corollaries of Theorem \ref{maintheorem} and    general facts about diffusion limits from
\cite{sv}, Chapter 11.
 
Let $\mathcal{H}$ be an open connected  subset of $\mathbb{R}^m$.
 We shall soon specialize to $\mathcal{H}=\mathbb{H}_-^m$ after reviewing  some background information.
 Let $\Omega$ be the space of continuous functions
 from $[0,\infty)$ to $\mathcal{H}$. Define $\pi_t:\Omega\rightarrow \mathcal{H}$ such that
 $\pi_t(\omega)=\omega(t)$. Then $\Omega$ has a natural metric topology making it a Polish space,  relative to which
 these  position maps are continuous. Let $\mathcal{M}$ be the Borel $\sigma$-algebra on $\Omega$, which is also the 
$\sigma$-algebra on $\Omega$  generated by all the $\pi_t$. Let $\mathcal{M}_t$ be the $\sigma$-algebra
on $\Omega$ generated by the $\pi_s$ such that $0\leq s\leq t$.

Now consider   a (time independent) second order elliptic differential operator  $\mathcal{L}$ with continuous coefficients acting  on compactly supported  smooth functions on $\mathcal{H}$. After \cite{sv}, a probability measure $\mathbb{P}$ on $(\Omega, \mathcal{M})$
is said to be a solution to the martingale problem for $\mathcal{L}$  starting  from $(s,v)\in [0,\infty)\times \mathcal{H}$
if the $\mathbb{P}$-probability  of the set of paths $\omega$ such that $\omega(t)=v$ for $0\leq t\leq s$ is $1$ and
$\varphi\circ \pi_t - \int_s^t (\mathcal{L}\varphi)\circ \pi_\tau\, d\tau$ is a $\mathbb{P}$-martingale after time $s$ for all
compactly supported smooth $\varphi$ on $\mathcal{H}$.

\begin{lemma}\label{functionphi}
A sufficient condition for the martingale problem to have exactly one solution 
is the existence of  (i) a non-negative function  $\varphi\in C^2(\mathcal{H})$ such that $\varphi(u_n)\rightarrow \infty$
as $u_n\rightarrow \infty$ (that is, $u_n$ eventually  leaves every compact set as $n\rightarrow \infty$)
and (ii) a constant $\lambda>0$  such that $\mathcal{L}\varphi\leq \lambda\varphi$. 
\end{lemma}
\begin{proof}
The  proof is easily extracted from the proof of Theorem 10.2.1, p. 254, of \cite{sv}.
\end{proof}

Given a family of  transition probabilities kernels  $u\mapsto \Pi_h(u,\cdot)$   with state space $\mathcal{H}$ parametrized by $h$, define
for each $v\in \mathcal{H}$  the family $\mathbb{P}^h_v$ of probability measures on $\Omega$ characterized by
the following properties (\cite{sv}, p. 267):
\begin{enumerate}
\item The set $\pi_0^{-1}(\{v\})$ has $\mathbb{P}^h_v$-probability $1$;
\item The set of polygonal paths $\omega$ such that
$$\omega(t)=\frac{(k+1)h-t}{h}\omega(kh) +\frac{t-kh}{h}\omega((k+1)h), $$
for all integer $k\geq 0$, has $\mathbb{P}^h_v$-probability $1$;
\item The conditional probability given $ \mathcal{M}_{kh}$ equals  $\Pi_h(\omega(kh),\cdot)$; that is,
$$\mathbb{P}^h_v(\pi_{(k+1)h}\in \Gamma \, |\,  \mathcal{M}_{kh})=\Pi_h(\pi_{kh}, \Gamma)$$ for all $k\geq 0$ and all $\Gamma$ in the Borel $\sigma$-algebra of $\mathcal{H}$.
\end{enumerate}

Conditions 1 and 2 mean that the distribution of $(\pi_{0}, \pi_h, \pi_{2h}, \dots)$ 
is the time-homogeneous Markov chain starting  from $v$ with transition probabilities $u\mapsto \Pi_h(u,\cdot)$.
Notice that we have used before the notation $\nu_u$ for $\Pi_h(u,\cdot)$.  Let $P_h$ be the corresponding operator
on compactly supported smooth functions and let $A_h:=P_h-I$. 
Condition 3 above is equivalent to: 
$$\varphi\circ\pi_{kh} -\sum_{j=0}^{k-1} (A_h\varphi)\circ\pi_{jh} \text{ is a } (\mathcal{M}_{kh}, \mathbb{P}^h_v)\text{-martingale}$$
 for  every  compactly supported smooth function $\varphi$ on $\mathcal{H}$.

The key fact we need from the general theory of diffusion processes can now be stated.
\begin{theorem}\label{maingeneral0}
Assume that (i) the elliptic second order differential operator  $\mathcal{L}$ (with continuous coefficients) is
such that for each $u\in \mathcal{H}$ there is a unique solution $\mathbb{P}_u$ to the martingale problem for
$\mathcal{L}$ starting at $u$; and (ii)
$ h^{-1} A_h\Phi $ converges to $ \mathcal{L}\Phi $ uniformly on compact sets 
for every smooth compactly supported $\Phi$   on $\mathcal{H}$. 
Then $\lim_{h\rightarrow 0}\mathbb{P}^h_u=\mathbb{P}_u$
and convergence is uniform in $u$ over compact subsets of $\mathcal{H}$.
\end{theorem}
\begin{proof}
A proof is easily adapted from that of  Theorem 11.2.3 of \cite{sv}.
\end{proof}

\subsection{Back to the random scattering  operators}
We now set $\mathcal{H}=\mathbb{H}_-^m$.
It was shown above that   for the  convergence of the Markov chain to a diffusion process it is sufficient to have:
 (i)    convergence of $h^{-1}A_h\Phi$ to $\mathcal{L}\Phi$ for every compactly supported smooth $\Phi$
as in Theorem \ref{maingeneral} and  (ii) a function $\varphi$ as in Lemma \ref{functionphi}.
The convergence required  in (i) is implied  by Theorem \ref{maintheorem}. We now show the existence of a  $\varphi$.

\begin{lemma}\label{lemmafunctionphi}
Let $\mathcal{L}$ be the differential operator of Theorem \ref{maintheorem}. Suppose that $\Lambda$ is adapted.
Then there is a smooth function $\varphi:\mathbb{H}_-^m\rightarrow (0,\infty)$ and a positive constant $\lambda$ such that
$\mathcal{L}\varphi\leq \lambda\varphi$ and $\varphi(v)\rightarrow \infty$ as $|v|\rightarrow \infty$ or $v$ approaches the
boundary of $\mathbb{H}_-^m$.  
\end{lemma}
\begin{proof}
As $\Lambda$ is adapted, we may assume that  $\mathcal{L}$ is as in Corollary \ref{kbig} ($k\geq 1$)
or as in Corollary \ref{constspeed} ($k=0$). The case $k=0$ is much simpler: take $\varphi(v)=c+|v|^2$
for a big enough constant $c$. So we assume   $\mathcal{L}$ is as in Corollary \ref{kbig}.
Let $u:(-\infty, 0)\rightarrow (0,1]$ be a smooth function such that $u(s)=1$ for $|s|\geq 1$ and
$u(s)=-s$ for $|s|\leq 1/2$.
 Now define
 $$\varphi(v)=c+8\text{Tr}(C\Lambda)+2\text{Tr}(\Lambda) + |v|^2 -\ln u(v_m) $$
 where $c$ is a positive constant still  to be chosen.
 It is clear that $\varphi(v)$ goes to infinity as claimed. 
 A straightforward computation shows 
 $$\mathcal{L}\varphi=\begin{cases} 8\text{Tr}(C\Lambda) + 2\text{Tr}(\Lambda) & \text{ if }  |v_m|\leq 1/2\\
8\text{Tr}(C\Lambda)  & \text{ if } |v_m|\geq 1
\end{cases}$$
If $1/2\leq |v_m|\leq 1$, the coefficients of $\mathcal{L}\varphi$ are seen to depend quadratically on $v_1, \dots, v_{m-1}$ and
are bounded in $v_m,  1/v_m$.  This shows that in this range of $v_m$ there are  constants $c, \lambda$ greater than $1$
such that $\mathcal{L}\varphi \leq c +\lambda |v|^2$. On the other ranges, $\mathcal{L}\varphi \leq \varphi$, in which $c=0$. 
\end{proof}

Thus we conclude:

\begin{theorem}\label{maingeneral}
The martingale problem for the random billiard differential operator $\mathcal{L}$ under the assumption that $\Lambda$ is
adapted has a unique solution $=\mathbb{P}_u$ for each $u\in \mathbb{H}_-^m$. 
Furthermore,
 $\lim_{h\rightarrow 0}\mathbb{P}^h_u=\mathbb{P}_u,$
 where $\mathbb{P}^h_u$ solves  the martingale problem for the  Markov chain with transition probabilities operator $P_h$. 
Convergence is uniform in $u$ over compact subsets of $\mathcal{H}$.
\end{theorem}
\begin{proof}
This follows from Theorem \ref{maingeneral0} and   Lemmas \ref{functionphi} and \ref{lemmafunctionphi}.
\end{proof}

It is useful to express the diffusion process with infinitesimal generator $\mathcal{L}$
as an stochastic differential equation.  

\begin{proposition}[It\^o SDE] We consider separately the cases $k=0$ and $k>0$. The operator  $\Lambda$ is
assumed positive definite on $\mathbb{R}^{m-1}$.
\begin{enumerate}
\item 
Under the conditions of Corollary \ref{constspeed} ($k=0$), the It\^o differential equation associated to the infinitesimal generator $\mathcal{L}$
has the form
$$dV_t= -4\Lambda V_t\, dt + \left[2\left(1-|V_t|^2\right) \Lambda\right]^{1/2}dB_t,$$ where $B_t$ is $n$-dimensional Brownian motion
restricted to the disc $D_1^n$. The Lebesgue measure on the disc is stationary for this process.
\item Under the conditions of Corollary \ref{kbig} ($k>0$), the same It\^o differential   equation
has the form
$$dV_t= Z(V_t)\, dt + b(V_t)\, dB_t,$$ where $B_t$ is $m$-dimensional Brownian motion restricted
to $\mathbb{H}_-^m$, $Z(v)$ is the vector field
$$Z(v):=-2\Lambda v + \left(\frac1{v_m}-\frac{v_m}{\sigma^2}\right) \left(\langle \Lambda v,v\rangle +\text{\em Tr}(C\Lambda)\right)e_m$$
and $b(v)$ is the linear map
$$ b(v)u:=v_m \Lambda^{1/2}u-\langle \Lambda^{1/2}v,u\rangle e_m +\text{\em Tr}\left(C\Lambda\right)^{1/2}u_m e_m. $$
The stationary distribution is Maxwell-Boltzmann, as described in Proposition \ref{invariantproposition}.
\end{enumerate}
\end{proposition}
\begin{proof}
This is an easy  exercise. The general relation between the infinitesimal generator of the diffusion  and
the It\^o equation can be found, for example, in \cite{ok}. When $k=0$, note that the drift term
always points into the disc as $\langle \Lambda v, v\rangle>0$ for all non-zero $v$. 
As $\mathcal{L}$ is a symmetric operator (see Theorem \ref{elliptic}),   $\int(\mathcal{L}\Phi)(v)\, d\mu(v)=0$ for all
compactly supported smooth $\Phi$, where $\mu$ is the stationary measure for $P$ given in Proposition
\ref{invariantproposition}.
The claim about the stationary distributions is a consequence of this property.
\end{proof}

\begin{proposition}\label{reductionLegendre}
Let $\mathcal{L}$ be as in Theorem \ref{maintheorem}.
Then the following two conditions are equivalent:
\begin{enumerate}
\item  $\mathcal{L}|v|^2=0$ 
\item  $\text{\em Tr}(C\Lambda)=0$ and $\, \text{\em Tr}(\Lambda)=
\text{\em Tr}(\Lambda^\curlyvee)$. 
\end{enumerate} If  these conditions hold,   the diffusion associated to $\mathcal{L}$
restricts to hemispheres of arbitrary radius (i.e., the level surfaces of $|v|^2$), and is equivalent  to a Legendre diffusion.
\end{proposition}
\begin{proof}
This follows from the observation that $\mathcal{L}|v|^2/4=2\text{Tr}(C\Lambda)+v_m\left(\text{Tr}\left(\Lambda^\curlyvee\right)-\text{Tr}(\Lambda)\right)$.
\end{proof}

The significance of this remark is the following.  In the examples, $\text{Tr}(\Lambda)-
\text{Tr}(\Lambda^\curlyvee)$ consists of mass ratios  
whose denominators are the masses associated to the velocity covariance matrix $C$. These
constitute the ``wall subsystem,'' whose kinetic variables are ``hidden.''
Therefore, the two conditions amount to the assumption that the masses of
the wall subsystem are infinite and have zero velocity. Thus we have an elastic random scattering system.

\subsection{Example: wall with particle structure}\label{multiplemasses}
Consider the idealized physical model  depicted in Figure \ref{multiple}.
It   consists of $k$ point masses $m_1, \dots, m_k$ that can slide without friction
on  the interval $[0,l]$ independently of each other,
 and a point mass $m$ that can similarly move on  the interval $[0, \infty)$.
When reaching the endpoints of $[0,l]$,   masses $m_i$ bounce off elastically, while $m$ collides elastically
with the $m_i$ but moves freely past  $z=l$.  
One may think of  the  $m_i$ are being  tethered to the left wall by  imaginary  (inelastic, massless and fully flexible) strings of length $l$;
when a string is fully extended, the corresponding mass bounces back as if due to a wall at $z=l$.
So the $m_i$  are restricted to $[0,l]$, but $m$ is free to move into this interval and may collide with the $m_i$. 

 \vspace{0.1in}
\begin{figure}[htbp]
\begin{center}
%\epsfile{file=bundle.eps,scale=0.8}
\includegraphics[width=2.5in]{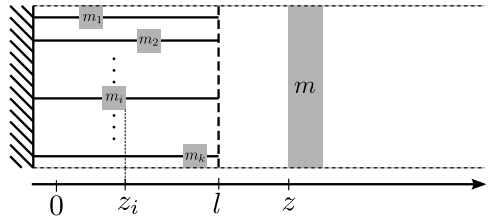}\ \ 
\caption{\small  In this model, masses $m_i$ constitute the {\em wall system}, while $m$ is the free mass. The system
is essentially one-dimensional.}
\label{multiple}
\end{center}
\end{figure} 

The positions of the $m_i$  are $z_i\in [0,l]$ and $z\in [0,\infty)$. Let $M=m+m_1+\dots +m_k$. In the new
 coordinates $x_i=\sqrt{m_i/M}\, z_i$, $x_{k+1}=\sqrt{m/M} z$,  the kinetic energy form  becomes
$$ K(x,\dot{x})=({M}/2) \left(\dot{x}_1^2+\cdots +\dot{x}_{k+1}^2\right).$$
We may equivalently assume that $(x_1, \dots, x_k)$ defines a point on the torus $\mathbb{T}^k$ by taking
the range of $x_i$ to be $[-a_i/2,a_i/2]$, where $a_i=2\sqrt{m_i/M}\, l$, and identifying the end points $ a_i/2$ and $-a_i/2$. Mass  $m$ is then constrained to move on the interval defined by
$$x_{k+1}\geq F(x_1, \dots, x_k):=\max\left\{\sqrt{m/m_1}\, |x_1|, \dots, \sqrt{m/m_{k}}\, |x_k|\right\}.$$
Thus the configuration manifold is $M=\left\{(x,x_{k+1})\in \mathbb{T}^k\times \mathbb{R}: x_{k+1}\geq F(x)\right\}$, 
and collision is represented (due to energy and momentum conservation and time-reversibility),  by
specular reflection at the boundary of $M$ as depicted in Figure \ref{tent}.

This deterministic billiard system can be turned into a random scattering system in
several ways. We illustrate two natural possibilities, which we call the {\em heat bath} model and
the {\em random elastic collision} model.  The assumptions for the heat bath model are as follows:   at the moment $m$ crosses
$z=l$ into $[0,l]$, the initial position of each $m_i$ is a random variable uniformly distributed over $[0,l]$,
and the velocity of $m_i$ is normally distributed with mean $0$ and variance $\sigma_i^2$.
We assume that these $\sigma_i^2$ are such that $m_i\sigma_i^2=m_j \sigma_j^2$ for all $i$ and $j$. In physical
terms, we are imposing a condition of equipartition of energy among the wall-bound masses. 
In the new coordinates $x_i$ the velocities are normal random variables with mean $0$ and equal variance
$\sigma^2=({m_i}/{M})\sigma_i^2$.

 \vspace{0.1in}
\begin{figure}[htbp]
\begin{center}
%\epsfile{file=bundle.eps,scale=0.8}
\includegraphics[width=2in]{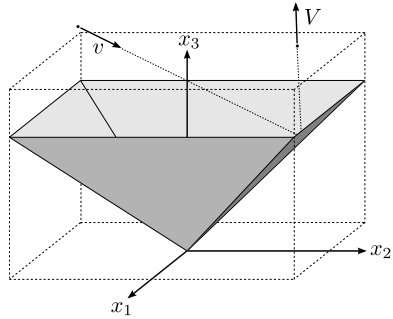}\ \ 
\caption{\small  Billiard representation of the system of Figure \ref{multiple} for $k=2$.}
\label{tent}
\end{center}
\end{figure}

This is  essentially the case described in Corollary \ref{cordim1}.
The dimensions for the heat bath model are: $n=k=m-1.$ Let $e_1, \dots, e_k, e=e_{k+1}$ be
the orthonormal basis of coordinate vector fields corresponding to the $x_i$. Then
the  projection to the hyperplane $e^{\perp}$
of the normal vector field $n(x)$ to the graph of $F$
   is
$$\overline{n}(x)=\pm \sqrt{\frac{m}{m+m_i}}\,  e_i \text{ for $x$ such that } |x_i|\geq \max_{j}|x_j|.$$
Therefore, $A$ is the diagonal matrix
$$A= V_0\sum_{i=1}^k \frac{m}{m+m_i}e_i^*\otimes e_i $$
where $V_0$ is the   volume of the sector $|x_i|\geq \max_j|x_j|$, normalized so that the total volume of the torus is $1$.
The matrix $C$ is the covariance matrix of the velocity $\dot{x}$, and is by assumption the scalar matrix
 $C=\sigma^2\sum_i e_i^*\otimes e_i$.

 Therefore,  
 $${\text{Tr}\left(CA\right)}/{\text{Tr}(A)}=\sigma^2,\ \text{Tr}(A)= \sum_{i=1}^k\frac{m}{m+m_i}, \ h=\max_j\left\{\frac{m}{m_1}, \dots, \frac{m}{m_k}\right\}.$$
Let us say for concreteness that all the wall-bound masses are equal to $m_0$, so $h=m/m_0$ and  $\Lambda$ becomes the
identity matrix times $V_0$, whose trace is $kV_0$.  From Corollary \ref{kbig} we obtain, for the heat bath model
with small ratio $m/m_0$ and $k$ equal bound masses,  the following differential operator. Let 
$v$ indicate the velocity of the free mass $m$ (in the new coordinate system, so $v=\dot{x}_{k+1}$)
and let $\Phi$ be any compactly supported smooth function on the interval $(0,\infty)$. Then 
\begin{equation}\label{laguerre}\mathcal{L}\Phi= 2kV_0\sigma^2\left[\left(\frac1v-\frac{v}{\sigma^2}\right)\Phi' +\Phi''\right]. \end{equation}
The corresponding It\^o diffusion has  the form
$$ dV_t=2kV_0\sigma^2\left(\frac1v-\frac{v}{\sigma^2}\right)\, dt + \sqrt{2kV_0\sigma^2}\, dB_t.$$
Figure \ref{OneD} shows a sample path for this SDE obtained by   Euler approximation.  (See \cite{kp}.)

We now consider the random elastic collision model. The assumptions for  this model are as follows:
the velocities of all the masses $m_i, m$ constitute the observable variables, and the positions
in $[0,l]$ of the wall-bound masses at the moment $m$ crosses into $[0,l]$ are uniformly distributed
random variables. This is the case to which Corollary \ref{constspeed} applies, where $n=k$. (The integer $k$ of
Theorem \ref{maintheorem} is  $0$.) Again, for concreteness, suppose that all the wall-bound masses are equal to $m_0$.
The eigenvalues of $\Lambda=V_0 I$ are all $\lambda_i=V_0$. We may assume without loss of generality
that the constant speed of the billiard particle (in the billiard representation of Figure \ref{tent}) is $1$ and let  $v$ denote
the projection of the billiard particle's velocity to the unit disc $D^k_1$ in dimension $k$.   Then  we obtain from Corollary \ref{constspeed}:
\begin{equation}\label{legendre11}(\mathcal{L}\Phi)(v) = 2 V_0 \sum_{i=1}^k\left(\left(1-|v|^2\right)\Phi_i\right)_i\end{equation}
where $\Phi$ is any compactly supported smooth function on  $D^k_1$.

The differential operator of \ref{legendre11}, as well as    \ref{genLegendreOp} in  Corollary \ref{constspeed}, generalize  in a natural way the standard
Legendre operator defined on functions of the interval $[-1,1]$.  We refer to
the diffusion process with this type of infinitesimal generator a (generalized) {\em Legendre diffusion}.
A sample path is shown in Figure \ref{legendre}.

It is interesting to notice that     the heat bath and random elastic collision models  lead to very standard Sturm-Liouville differential operators. For the heat bath, Equation \ref{laguerre}
is, up to constant, Laguerre's differential operator.
Essentially the same  model of heat bath/thermostat described here is used in \cite{chum}
to build a minimalist mathematical model of a heat engine, described as  a random system of billiard type.

 \vspace{0.1in}
\begin{figure}[htbp]
\begin{center}
%\epsfile{file=bundle.eps,scale=0.8}
\includegraphics[width=3.5in]{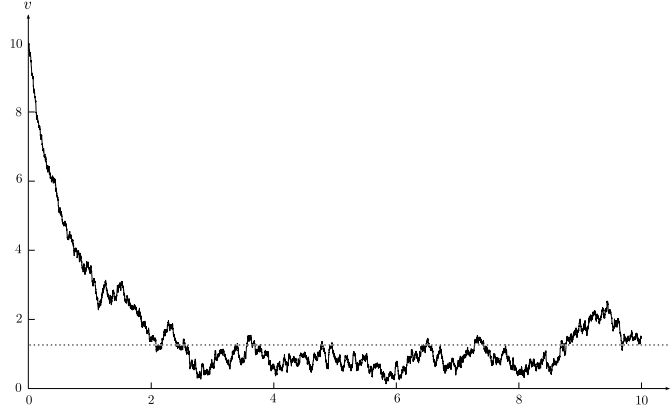}\ \ 
\caption{\small 
A sample path of the It\^o equation
$dV_t=(1/v - v)dt +   dB_t$. We have used Euler approximation with time length 
$10$, initial point $v=10$, and  number of steps $=10000$. The mean value relative to the stationary distribution is $\sqrt{\pi/2}$,
corresponding to the horizontal dashed line.}
\label{OneD}
\end{center}
\end{figure}

\subsection{Example: collisions of point mass and moving surface}
Here we give the differential equation approximation of  $P-I$ for the
example of Subsection \ref{movingsurface}. 
The main interest in this example is that it is the simplest that combines the features of
the two cases considered above in Subsection \ref{multiplemasses}.

It is first necessary to describe 
  the operator $A$ (see Subsection \ref{appargument}). The notations are as in that subsection.
  The torus $\mathbb{T}^2$ has fundamental domain (centered at $(0,0)$) 
$$ |x_0|\leq \tau/2,\ \ |x_1|\leq 1/2, \ \ \tau = \frac{a_0}{a_1}\sqrt{\frac{m_0}{m_1}}.$$
The billiard boundary surface is the graph of  $x_2=F(x_0, x_1)$, 
where
$$ F(x_0,x_1)= \sqrt{\frac{m_1}{m_0}}|x_0| +a_1^{-1}f(a_1x_1).$$
Let $e_0, e_1, e:=e_2$ be the standard coordinate vector fields for the coordinate system $(x_0,x_1, x_2)$.
It is easily checked that $A$ is
$$ A=\left(\frac{m_1}{m_0}+ O(h^4)\right)e_0^*\otimes e_0+ \left(\int_0^1 [f'(a_1 s)]^2\, ds+ O(h^4)\right) e_1^*\otimes e_1,$$
where the error term satisfies $0\leq O(h^2)\leq h^2$, while the norm of $A$ 
satisfies $\|A\|\leq h.$
As an  example, take $$f(z_1)=\sqrt{R^2-z_1^2}-\sqrt{R^2-a_1^2/4}.$$ The graph of $f$  is an arc of circle 
of radius $R$ intersecting the $z_1$-axis at the points $(\pm a_1/2, 0)$. Let the {\em scale-free curvature}
be $\kappa:=a_1/R<1$. Then 
$$a:=\int_0^1[f'(a_1s)]^2\, ds ={\kappa}^{-1} \ln\frac{1+\frac{\kappa}{2}}{1-\frac{\kappa}{2}} -1=\frac{\kappa^2}{12}+O(\kappa^3).$$
Thus for small values of $h$   (disregarding terms in  $\kappa$ or order greater
than $2$, and in $m_1/m_0$ of order greater than $1$)  we have
$$A=\frac{m_1}{m_0} e_0^*\otimes e_0+ \frac{\kappa^2}{12}e_1^*\otimes e_1 \text{ and } h= \frac{\kappa^2}{4}+\frac{m_1}{m_0}. $$
The operator $C$   takes the form
$$C=\int_{-\infty}^{\infty}w^2\, d\mu(w)e_0^*\otimes e_0= \sigma^2 e_0^*\otimes e_0.$$
We observe  that $\text{Tr}(CA)=\frac{m_1}{m_0}\sigma^2 $ and $\text{Tr}(A)=\frac{m_1}{m_0} +a$.
For  the special case of an arc of circle, 
$a=\kappa^2/12$, where $\kappa$ is the scale free curvature.
This yields the   approximation, written informally as
 \begin{equation}\label{updownlaplacian}
\frac{(P\Phi)(v)-\Phi(v)}{2}\approx  \frac{m_1}{m_0}\sigma^2\mathcal{L}_{\text{\tiny temp}}\Phi+   \frac{\kappa^2}{12}\mathcal{L}_{\text{\tiny curv}}\Phi
\end{equation}
where 
 \begin{align*}
\mathcal{L}_{\text{\tiny temp}}\Phi&=  \left(\frac{1}{v_2}-\frac{v_2}{\sigma^2}\right)\Phi_2 +\Phi_{22} \\
\mathcal{L}_{\text{\tiny curv}}\Phi &=
 -2 v_1\Phi_1+\frac{v_1^2-v_2^2}{v_2}\Phi_2 -2v_1v_2 \Phi_{12} +v_2^2\Phi_{11} + v_1^2 \Phi_{22}
\end{align*}

The mass-ratio and curvature parameters may
{\em a priori}   go to $0$ independently with $h$ (under
Assumption \ref{assumption1}) and  the
particular way in which each goes to $0$  matters for  the limit. Expression \ref{updownlaplacian}
shows that taking  $h$ for  the denominator in the quotient  $(P_h\Phi-\Phi)/h$ used in the definition
of $\mathcal{L}$ is  essentially an arbitrary choice. One could have taken instead $\text{Tr}(A)$, for example. 
If  we further ask  in  this example that
  the scale-free curvature and the mass ratio be coupled by a linear relation such as $\frac{m_1}{m_0}=\alpha\frac{\kappa^2}{4}$,
  for a fixed  but arbitrary constant $\alpha>0$, and keep 
 the original choice of denominator $h$,  then
  $$\Lambda =\lim_{h\rightarrow 0}A/h= \frac{\alpha}{1+\alpha}e_0^*\otimes e_0 + \frac1{3(1+\alpha)}e_1^*\otimes e_1.$$
This gives the family of operators (depending on $\alpha$)
 \begin{align*}
\left( \mathcal{L}\Phi\right)(v)=\lim_{h\rightarrow 0}
\frac{(P_h\Phi)(v)-\Phi(v)}{h}&=
\frac{2\sigma^2\alpha }{1+\alpha}
 \left\{\left(\frac{1}{v_2}-\frac{v_2}{\sigma^2}\right)\Phi_2 +\Phi_{22}\right\} + \\
 &
\frac2{3(1+\alpha)}\left\{-2 v_1\Phi_1+\frac{v_1^2-v_2^2}{v_2}\Phi_2 -2v_1v_2 \Phi_{12} +v_2^2\Phi_{11}  + v_1^2 \Phi_{22}\right\}.
\end{align*}
In the concrete example  of Figure \ref{example3SDE} we took $\alpha=1$, $\sigma^2=1/3$ (and multiply the
the operator by an overall factor $3$ to make it look simpler).
  
The expression  \ref{updownlaplacian},
points to  a separation between, on the one hand,  the term responsible 
for the change in speed, which    contains the variance (temperature) $\sigma^2$ and the mass ratio,
and on the other, a purely geometric term that involves the scale free curvature $\kappa$. 
If we let $\sigma^2$ be  $0$ and the wall mass   $\infty$, and consider  $\kappa$ to be small,
the MB-Laplacian reduces to $\mathcal{L}_{\text{\tiny curv}}$.
It is interesting to note that  $\mathcal{L}_{\text{\tiny curv}}|v|^2=0$, so the diffusion
associated to this second order operator restricts to hemispheres of arbitrary radius,
and we have a Legendre diffusion. (See Proposition \ref{reductionLegendre}.)
 
\vspace{0.2in}
{\it Acknowledgment:} Hong-Kun Zhang was partially funded by NSF  Grant DMS-090144 and NSF CAREER Grant DMS-1151762.

\end{document}